\documentclass[11pt]{article}
\usepackage[utf8]{inputenc}
\usepackage[margin=1in]{geometry}
\usepackage{fullpage}

\usepackage{amsmath,amssymb,amsthm,mathtools}
\usepackage[nottoc]{tocbibind} 

\usepackage{thm-restate}
\usepackage{xspace,bm,floatrow,comment}

\usepackage{enumitem}
\setlist{listparindent=\parindent,parsep=0pt,itemsep=1em}
\setlist[itemize]{label=$-$,noitemsep}
\setlist[enumerate]{itemsep=1mm}
\setlist[description]{leftmargin=\parindent}
\usepackage[normalem]{ulem}

\usepackage{hyperref}
\usepackage[dvipsnames]{xcolor}
\definecolor{newcolor}{hsb}{0.6,1,0.75}
\hypersetup{
	colorlinks=true,
	urlcolor=newcolor,
	linkcolor=newcolor,
	citecolor=newcolor,
	unicode
}
\usepackage{physics}

\definecolor{betterYellow}{RGB}{255,255,0}

\usepackage[capitalise,nameinlink,noabbrev]{cleveref}
\usepackage{algpseudocode, graphicx, times, thmtools, thm-restate}
\usepackage[ruled,vlined,linesnumbered]{algorithm2e}
\SetKwProg{Procedure}{Procedure}{}{}
\MakeRobust{\Call}

\usepackage[noadjust]{cite}
\usepackage{tocloft}
\usepackage{tikz}
\usepackage{graphicx}
\usetikzlibrary{positioning,arrows.meta,math,shapes.geometric,decorations.pathmorphing,decorations.pathreplacing}

\usepackage[font=small,labelfont=bf]{caption} 
\usepackage{subcaption}
\usepackage[percent]{overpic}

\usepackage[english]{babel}
\addto\extrasenglish{}
\addto\extrasenglish{}

\setlist[enumerate]{noitemsep}
\setlist[itemize]{label=$-$,noitemsep}

\makeatletter
\DeclareRobustCommand\bfseries{%
  \not@math@alphabet\bfseries\mathbf
  \fontseries\bfdefault\selectfont\boldmath}
\makeatother

\theoremstyle{definition}
\declaretheorem[name=Theorem, numberwithin=section]{theorem}

\declaretheorem[name=Lemma,sibling=theorem]{lemma}
\declaretheorem[name=Corollary, sibling=theorem]{corollary}
\declaretheorem[name=Claim,sibling=theorem]{claim}
\declaretheorem[name=Definition, sibling=theorem, style=definition]{definition}
\declaretheorem[name=Fact,sibling=theorem]{fact}
\declaretheorem[name=Observation,sibling=theorem]{observation}

\usepackage[T1]{fontenc}

\newcommand{\sens}{\mathsf{Sens}}
\newcommand{\clo}{\mathsf{Clo}}
\newcommand{\swapsens}{\mathsf{SwapSens}}
\newcommand{\swapdist}{\mathsf{SwapDist}}
\newcommand{\swapclo}{\mathsf{SwapClo}}
\newcommand{\TV}{\mathsf{TV}}
\newcommand{\EMD}{\mathsf{EMD}}
\newcommand{\Ham}{\mathsf{Ham}}

\newcommand{\opt}{\mathsf{opt}}
\newcommand{\cost}{\mathsf{cost}}
\newcommand{\val}{\mathsf{val}}

\newcommand{\ynote}[1]{\textcolor{red}{(Yuichi: #1)}}

\newcommand{\E}{\mathop{\mathbf{E}}}

\allowdisplaybreaks

\title{Sensitivity Lower Bounds for Approximation Algorithms}
\author{Noah Fleming\thanks{Research supported by NSERC.} \\
Lund University \& Columbia University \\
\texttt{noah.fleming@cs.lth.se}
\and
Yuichi Yoshida\thanks{Supported by JSPS KAKENHI Grant Number 24K02903.} \\
National Institute of Informatics \\
\texttt{yyoshida@nii.ac.jp}
}
\RequirePackage[normalem]{ulem} 
\RequirePackage{color}\definecolor{RED}{rgb}{1,0,0}\definecolor{BLUE}{rgb}{0,0,1} 
\providecommand{\DIFaddbegin}{} 
\providecommand{\DIFaddend}{} 
\providecommand{\DIFdelbegin}{} 
\providecommand{\DIFdelend}{} 
\providecommand{\DIFaddbeginFL}{} 
\providecommand{\DIFaddendFL}{} 
\providecommand{\DIFdelbeginFL}{} 
\providecommand{\DIFdelendFL}{} 
\newcommand{\DIFscaledelfig}{0.5}
\RequirePackage{settobox} 
\RequirePackage{letltxmacro} 
\newsavebox{\DIFdelgraphicsbox} 
\newlength{\DIFdelgraphicswidth} 
\newlength{\DIFdelgraphicsheight} 
\LetLtxMacro{\DIFOincludegraphics}{\includegraphics} 
\newcommand{\DIFaddincludegraphics}[2][]{{\color{blue}\fbox{\DIFOincludegraphics[#1]{#2}}}} 
\newcommand{\DIFdelincludegraphics}[2][]{
\sbox{\DIFdelgraphicsbox}{\DIFOincludegraphics[#1]{#2}}
\settoboxwidth{\DIFdelgraphicswidth}{\DIFdelgraphicsbox} 
\settoboxtotalheight{\DIFdelgraphicsheight}{\DIFdelgraphicsbox} 
\scalebox{\DIFscaledelfig}{
\parbox[b]{\DIFdelgraphicswidth}{\usebox{\DIFdelgraphicsbox}\\[-\baselineskip] \rule{\DIFdelgraphicswidth}{0em}}\llap{\resizebox{\DIFdelgraphicswidth}{\DIFdelgraphicsheight}{
\setlength{\unitlength}{\DIFdelgraphicswidth}
\begin{picture}(1,1)
\thicklines\linethickness{2pt} 
{\color[rgb]{1,0,0}\put(0,0){\framebox(1,1){}}}
{\color[rgb]{1,0,0}\put(0,0){\line( 1,1){1}}}
{\color[rgb]{1,0,0}\put(0,1){\line(1,-1){1}}}
\end{picture}
}\hspace*{3pt}}} 
} 
\LetLtxMacro{\DIFOaddbegin}{\DIFaddbegin} 
\LetLtxMacro{\DIFOaddend}{\DIFaddend} 
\LetLtxMacro{\DIFOdelbegin}{\DIFdelbegin} 
\LetLtxMacro{\DIFOdelend}{\DIFdelend} 
\DeclareRobustCommand{\DIFaddbegin}{\DIFOaddbegin \let\includegraphics\DIFaddincludegraphics} 
\DeclareRobustCommand{\DIFaddend}{\DIFOaddend \let\includegraphics\DIFOincludegraphics} 
\DeclareRobustCommand{\DIFdelbegin}{\DIFOdelbegin \let\includegraphics\DIFdelincludegraphics} 
\DeclareRobustCommand{\DIFdelend}{\DIFOaddend \let\includegraphics\DIFOincludegraphics} 
\LetLtxMacro{\DIFOaddbeginFL}{\DIFaddbeginFL} 
\LetLtxMacro{\DIFOaddendFL}{\DIFaddendFL} 
\LetLtxMacro{\DIFOdelbeginFL}{\DIFdelbeginFL} 
\LetLtxMacro{\DIFOdelendFL}{\DIFdelendFL} 
\DeclareRobustCommand{\DIFaddbeginFL}{\DIFOaddbeginFL \let\includegraphics\DIFaddincludegraphics} 
\DeclareRobustCommand{\DIFaddendFL}{\DIFOaddendFL \let\includegraphics\DIFOincludegraphics} 
\DeclareRobustCommand{\DIFdelbeginFL}{\DIFOdelbeginFL \let\includegraphics\DIFdelincludegraphics} 
\DeclareRobustCommand{\DIFdelendFL}{\DIFOaddendFL \let\includegraphics\DIFOincludegraphics} 
\makeatletter 
\let\sout@orig\sout 
\renewcommand{\sout}[1]{\ifmmode\text{\sout@orig{\ensuremath{#1}}}\else\sout@orig{#1}\fi} 
\makeatother 
\RequirePackage{listings} 
\RequirePackage{color} 
\lstdefinelanguage{DIFcode}{ 
  moredelim=[il][\color{red}\sout]{\%DIF\ <\ }, 
  moredelim=[il][\color{blue}\uwave]{\%DIF\ >\ } 
} 
\lstdefinestyle{DIFverbatimstyle}{ 
	language=DIFcode, 
	basicstyle=\ttfamily, 
	columns=fullflexible, 
	keepspaces=true 
} 
\lstnewenvironment{DIFverbatim}{\lstset{style=DIFverbatimstyle}}{} 
\lstnewenvironment{DIFverbatim*}{\lstset{style=DIFverbatimstyle,showspaces=true}}{} 
\begin{document}
\maketitle
\begin{abstract}
    Sensitivity measures how much the output of an algorithm changes, in terms of Hamming distance, when part of the input is modified.
    While approximation algorithms with low sensitivity have been developed for many problems, no sensitivity lower bounds were previously known for approximation algorithms.
    In this work, we establish the first polynomial lower bound on the sensitivity of (randomized) approximation algorithms for constraint satisfaction problems (CSPs) by adapting the probabilistically checkable proof (PCP) framework to preserve sensitivity lower bounds. From this, we derive polynomial sensitivity lower bounds for approximation algorithms for a variety of problems, including maximum clique, minimum vertex cover, and maximum cut.

    Leveraging the connection between sensitivity and locality in the non-signaling model, which subsumes the $\mathsf{LOCAL}$, quantum-$\mathsf{LOCAL}$, and bounded dependence models, we establish locality lower bounds for several graph problems in the non-signaling model.
\end{abstract}
\thispagestyle{empty}

\newpage
\tableofcontents
\thispagestyle{empty}

\setcounter{page}{0}
\newpage

\section{Introduction}

The notion of algorithmic sensitivity, introduced by Varma and Yoshida~\cite{varma2023average}, measures the stability of an algorithm's output in Hamming distance (or the earth mover's distance if the algorithm is randomized) when an element in the input is added or deleted. 
In practical situations, low sensitivity is desirable because inputs can easily change due to noise or over time. If an algorithm has high sensitivity, it may lose reproducibility, erode user trust, and lead to inconsistent decisions. 
Since this concept was introduced, low-sensitivity algorithms have been developed for many graph problems.
Some representative results include a $2$-approximation algorithm for the minimum vertex cover problem with sensitivity $O(1)$~\cite{censor2016optimal,varma2023average}, a $(1-\varepsilon)$-approximation algorithm for the maximum matching problem with sensitivity $\Delta^{O(1/\varepsilon)}$~\cite{yoshida2026personal}, where $\Delta$ is the maximum degree, an additive $O(n^{2/3})$-approximation algorithm for the minimum $s$-$t$ cut problem with sensitivity $O(n^{2/3})$, where the output is a vertex set~\cite{varma2023average}, and a $(1+\varepsilon)$-approximation algorithm for the shortest path problem with sensitivity $O(\varepsilon^{-1} \log^3 n)$ (with respect to edge contractions instead of additions or deletions)~\cite{kumabe2023lipschitz}.
Moreover, it has been shown that algorithms with low sensitivity can be used to design online algorithms with small recourse~\cite{yoshida2022average}, as well as online learning algorithms with low regret~\cite{dong2024batch}.
Therefore, it is important to understand whether low sensitivity algorithms can be achieved for each problem.

On the lower bound side, the only known general result that we are aware of is that any randomized algorithm for finding a proper $2$-coloring of a bipartite graph with $n$ vertices requires sensitivity $\Omega(n)$~\cite{varma2023average}, which follows easily from the fact that a connected bipartite graph has only two proper $2$-colorings such that the Hamming distance between them is $\Omega(n)$.
Indeed, there are no known lower bounds on the sensitivity of approximation algorithms, leaving a significant gap in our knowledge. This is especially true, since all of the known upper bounds are derived from approximation algorithms.

\subsection{Our Contributions}
\subsubsection{Sensitivity Lower Bounds}
In this work, we show polynomial sensitivity lower bounds for (even randomized) approximation algorithms for various combinatorial problems by adapting the probabilistically checkable proofs (PCP) framework~\cite{arora1998proof,arora1998probabilistic}.
We emphasize that this is the first time that lower bounds for \emph{approximation} algorithms have been established.
We are able to derive lower bounds for a variety of combinatorial problems through reductions, and we present some representative results next. 

We start with the maximum clique problem.
Consider a (potentially inefficient) $n^{-\varepsilon}$-approximation algorithm that outputs a clique of size at most $n^{1-\varepsilon}$.
Its sensitivity is trivially bounded by $O(n^{1-\varepsilon})$ and this is the only known upper bound.
We show that the polynomial dependency on $n$ is necessary:
\begin{theorem}\label{thm:max-clique-intro}
     There are universal constants $\varepsilon, \delta > 0$ such that any (inefficient) algorithm for the maximum clique problem that outputs an $n^{-\varepsilon}$-approximate clique with probability $1-O(1/n)$ has sensitivity $\Omega(n^\delta)$.
\end{theorem}
We note that the lower bound also applies to the maximum independent set problem.
Additionally, our proof is information-theoretic, and the lower bounds presented in this work apply even to inefficient algorithms.

Next, we consider the minimum vertex cover problem.
As we mentioned, there exists a $2$-approximation algorithm for the minimum vertex cover problem with sensitivity $O(1)$~\cite{censor2016optimal,varma2023average}. 
We prove that the sensitivity must increase significantly as the approximation ratio approaches one.
\begin{theorem}\label{thm:min-vertex-cover-intro}
    There are universal constants $\varepsilon, \delta > 0$ such that any (randomized and inefficient)  $(1+\varepsilon)$-approximation algorithm for the minimum vertex cover problem has sensitivity $\Omega(n^\delta)$.
\end{theorem}

Finally, we discuss the maximum cut problem, where the output is a vertex set.
Trivially, the random assignment algorithm is a $1/2$-approximation algorithm with sensitivity zero, and we can obtain an (inefficient) additive $\widetilde O(\sqrt{\varepsilon^{-1}nm})$-approximation algorithm with sensitivity $O(\varepsilon n)$ by converting a differentially private algorithm~\cite{eliavs2020differentially}.
As mentioned, the only known lower bound is that exact algorithms, which can solve $2$-coloring,  must have sensitivity $\Omega(n)$~\cite{varma2023average}. 
We show that the sensitivity must remain polynomial as the approximation ratio approaches one:
\begin{theorem}\label{thm:max-cut-intro}
    There exist universal constants $\varepsilon, \delta > 0$ such that any (randomized and inefficient) $(1-\varepsilon)$-approximation algorithm for the maximum cut problem has sensitivity $\Omega(n^\delta)$. 
\end{theorem}
We can obtain similar lower bounds for other constraint satisfaction problems (CSPs) including $\mathsf{E3SAT}$ and $\mathsf{3LIN}$.

\subsubsection{Locality Lower Bounds for the Non-Signaling Model}

Although the notion of sensitivity is interesting in its own right, it has a close connection to the \emph{non-signaling model}~\cite{akbari2025online,coiteux2024no}, also known as the $\varphi$-$\mathsf{LOCAL}$ model and the causal model~\cite{arfaoui2014can,gavoille2009can}.
In this model, given a graph $G=(V,E)$, we can produce an arbitrary output distribution  as long as it does not violate the non-signaling principle: for any set of vertices $S \subseteq V$, modifying the structure of the input graph at more than a distance $t$ from $S$ does not affect the output distribution of $S$.
The parameter $t$ is called the \emph{locality} of the model.
It is stronger than any ``physical'' synchronous distributed computing model, and in particular, the non-signaling model is more powerful than the $\mathsf{LOCAL}$ model for distributed algorithms on graphs~\cite{linial1992locality}, the quantum-$\mathsf{LOCAL}$ model~\cite{gavoille2009can}, and the bounded dependence model~\cite{holroyd2017finitary}, all with the same locality~\cite{akbari2025online}.

Suppose that we can generate a distribution in the non-signaling model for a graph problem with locality $t$.
Then, an algorithm that samples from this distribution and outputs the result has sensitivity $O(\Delta^t)$, where $\Delta$ is the maximum degree of the graph, since deleting an edge affects the marginal distributions of at most $\Delta^t$ vertices.
Building on this connection, we can extend various lower bounds that were only known for the $\mathsf{LOCAL}$ model to the non-signaling model, including an $\Omega(1/\varepsilon)$ lower bound for an $n^{-\varepsilon}$-approximation algorithm for the maximum independent set problem~\cite{bodlaender2016brief}, an $\Omega(\log n)$ lower bound for a $(1+\varepsilon)$-approximation algorithm for the minimum vertex cover problem for some constant $\varepsilon>0$~\cite{goos2014no,faour2022distributed}, and an $\Omega(\log n)$ lower bound for a $(1-\varepsilon)$-approximation algorithm for the maximum cut problem for some constant $\varepsilon>0$~\cite{chang2023complexity}.
We note that our lower bounds match the known upper bounds in the $\mathsf{LOCAL}$ model~\cite{chang2023complexity}, and hence they imply that the non-signaling model offers no additional advantage for these problems.

\subsection{Proof Overview}\label{subsec:proof-overview}

An instance of a \emph{constraint satisfaction problem (CSP)} is a tuple $(V, E, \Sigma, \mathcal{R} = \{R_e\}_{e \in E})$, where $(V, E)$ is a hypergraph, $\Sigma$ is a finite domain (often referred to as the \emph{alphabet}), and $R_e$ is a $|e|$-ary relation over $\Sigma$, i.e., $R_e \subseteq \Sigma^{|e|}$. 
We say that an assignment $\sigma: V \to \Sigma$ \emph{satisfies} a constraint on $e = (v_1, \dots, v_k)$ if $(\sigma(v_1), \dots, \sigma(v_k)) \in R_e$.
For $1 \geq c \geq s \geq 0$, we define $\mathsf{MaxCSP}_{c,s}$ as the problem where, given a $c$-satisfiable CSP instance, i.e., there exists an assignment that satisfies at least a $c$-fraction of the constraints, the goal is to compute an $s$-satisfying assignment for the instance, i.e., an assignment that satisfies at least an $s$-fraction of the constraints.
Here, $c$ and $s$ are referred to as the \emph{completeness} and \emph{soundness} parameters.

We first present a sensitivity lower bound for the case where the gap between completeness and soundness is extremely small. 
Specifically, we show that any algorithm for $\mathsf{MaxCSP}_{1, 1-\Theta(1/n)}$ requires sensitivity of $\Omega(n)$ for $\mathsf{E2LIN}$ on cycles, where each constraint is of the form $x + y = 0 \pmod 2$ or $x + y = 1 \pmod 2$.

Next, we gradually increase the gap in multiple rounds without reducing the sensitivity lower bound.
To do so, we adapt the probabilistically checkable proof (PCP) framework~\cite{arora1998probabilistic,arora1998proof}, which was originally developed as a characterization of $\mathsf{NP}$ and have been useful for proving inapproximability results. 
We note that we are proving hardness for $\mathsf{NP}$-hard problems by reducing from a polynomial-time tractable problem, which might seem counterintuitive. 
This is possible because we are analyzing the sensitivity of the problems from an information-theoretic perspective, rather than focusing on their time complexity.

We build on Dinur's PCP construction~\cite{dinur2007pcp}, which consists of four steps: degree reduction, expanderization, gap amplification, and alphabet reduction. 
Below, we explain each step and describe how we modify them to ensure the sensitivity lower bound remains (almost) intact.
Since the argument for gap amplification and reduction is already covered in Dinur's original proof, we will primarily focus on how to analyze the increase and decrease in sensitivity.

We first note that the reduction in each step generally works as follows:
Let $\mathcal I$ be a family of instances. 
Consider converting a CSP instance $I = (V,E,\Sigma,\mathcal R)\in \mathcal I$ to another instance $I' = (V',E',\Sigma',\mathcal R')$. 
Suppose we have an algorithm $A'$ for the latter. 
We then run $A'$ on $I'$ to obtain an assignment $\sigma'$, and subsequently convert $\sigma' : V' \to \Sigma'$ back into an assignment $\sigma : V \to \Sigma$ for $I$.
Let $\mathcal I'$ be the family of instances obtained from instances in $\mathcal I$ by converting them. 
The sensitivity of algorithms for $\mathcal I'$ is bounded from below by the sensitivity for $\mathcal I'$, multiplied by two parameters, $C_I$ and $C_\sigma$, which are determined by the conversion of instances and assignments, respectively. 
Specifically, $C_I$ represents the number of changes to the new instance $I'$ when a constraint in the original instance $I$ is changed, and $C_\sigma$ represents the number of changes to the assignment $\sigma$ for $I$ when a value in the assignment $\sigma'$ for $I'$ is altered.
To maintain a large sensitivity lower bound, $C_I$ and $C_\sigma$ need to be small.

We note that known PCP constructions may have large $C_I$ and $C_\sigma$ (indeed, Dinur's original construction has large $C_I$ and $C_\sigma$). 
Therefore, we need to carefully modify both the construction and the procedure that converts $\sigma'$ to $\sigma$.
In particular, we cannot use known PCP constructions as black boxes when proving sensitivity lower bounds. As well, strategies such as taking the union of an instance $I$ which has high sensitivity with an instance $J$ which has a large gap (such as the instance obtained by running the PCP on $I$) on disjoint variables do not work. 
We address this in \Cref{sec:appendix}. 

The proof of Dinur's PCP construction proceeds by iteratively applying the following four steps; we now sketch how they can be modified in order to preserve sensitivity.
\paragraph{Degree Reduction.}
The goal of this step is to reduce the degree of each variable to a small constant and make the underlying graph regular. 
The reduction is straightforward: for each variable $v \in V$ with degree $d$ in the original instance $I$, we split $v$ into $d$ copies, $v_1, \dots, v_d$, and connect them with an expander, adding equality constraints on the newly introduced edges. 
To obtain a label for $v$ in the original instance $I$, we select one of $v_1, \dots, v_d$ uniformly at random and use the label of the chosen $v_i$ in $I'$.

A key feature of this transformation is that it actually \emph{increases} the sensitivity lower bound by $d$, which is crucial to offset the sensitivity decrease in other steps.

\paragraph{Expanderization.}
In this step, we simply superimpose an expander with a trivial constraint (satisfied by any assignment) onto the instance $I$ to transform the underlying graph into an expander, which is important for the gap amplification step. 
We use the assignment for $I'$ directly as the assignment for $I$.

\paragraph{Gap amplification.}
The goal of this step is to increase the gap between completeness and soundness by a constant factor $t$. 
Suppose the underlying graph is a $d$-regular expander. 
To achieve this, we sample a vertex $u \in V$, perform a random walk of length $O(t)$, and reach vertex $v$.
We then add a constraint between $u$ and $v$ over the alphabet $\Sigma^{1+d+\cdots+d^t}$, which encodes the assignments of the $t$-hop neighborhoods, checking whether the labels of the variables along the path satisfy all the constraints.
(This only defines a distribution over constraints; we convert this into an unweighted instance by multiplying the probability mass of each constraints by $d^{O(t)}$ to compute the number of copies needed for the constraint.)

To recover an assignment $\sigma$ for $I$ from an assignment $\sigma'$ for $I'$, for each variable $u \in V$, we gather the labels of variables in the $t$-hop neighborhood via a random walk. 
In Dinur's original proof~\cite{dinur2007pcp}, the majority of these labels would be used as the label for $u$ in $I$. 
However, this recovery procedure does not allow us to preserve the sensitivity lower bound, as changing a single label could change the majority for many variables. 
To smooth this transition, we would like to choose a label from a distribution weighted according to the frequency that it occurs in this random walk. 
However, the possibility of returning a label with low probability mass ruins the soundness argument. 
Instead, we consider a distribution only over labels with sufficiently large probability mass. 
Combined with a careful conditioning argument, this allows us to preserve the sensitivity lower bound while still increasing the gap.

\paragraph{Alphabet reduction.} 
This step decreases the size of the alphabet to a small constant. 
To do so, while preserving the gap, we take the Hadamard encoding of our alphabet. 
However, doing this converts our graph into a hypergraph, as every constraint now depends on many Boolean variables. 
To remedy this, we instead want to check whether a given assignment is \emph{close} to a satisfying assignment to that constraint by examining only a few bits of that assignment.
To do so, we apply a PCP (called an assignment tester~\cite{dinur2007pcp} or a PCP of proximity~\cite{ben2004robust}) to this constraint: 
we introduce sets of variables which purportedly encode the Hadamard table $L$ of a satisfying assignment $\alpha$ for our constraint, as well as variables for the Hadamard table $Q$ of $\alpha \otimes \alpha$. 
We then tests (via constraints) whether this is indeed the case, using applications of the BLR linearity test~\cite{BlumLR93}. 

To recover an assignment $\sigma$ for $I$ from an assignment $\sigma'$ for $I'$, for each variable $u \in V$, we consider the bits $\sigma'[u]$ which purportedly contain a label for $u$ --- a Hadamard codeword.
We would like to map $\sigma'[u]$ to the closest Hadamard codeword (and hence label to $u$), however this would not allow us to preserve our bound on sensitivity: consider any assignments to $\sigma'[u]$ which lies on the unique decoding radius of a codeword.
Then, changing a single bit in this assignment could change which codeword $\sigma'[u]$ is mapped to, and hence which label $\sigma$ assigns to $u$.
Instead, we smooth this transition by employing a randomized thresholding argument: We choose a random threshold $\tau \in [0,1/4]$ (as one may recover uniquely a Hadamard codeword which has been $1/4$ corrupted) and assign $\sigma'[u]$ to its closest codeword if its relative distance to that codeword is at most $\tau$, and to an arbitrary but fixed codeword otherwise. 
This smooth transition, along with a careful conditioning argument, allows us to maintain our bound on sensitivity.

Finally, we note that although we can directly apply the assignment tester to the original $\mathsf{E2LIN}$ instances to obtain a soundness of $1-\varepsilon$, the instance size grows at least exponentially. As a result, the sensitivity lower bound obtained this way is $\Omega(\log n)$, which is much weaker than our lower bound of $\Omega(n^\delta)$.

\subsection{Discussions}\label{subsec:discussions}

This work presents the first sensitivity lower bounds for approximation algorithms and establishes lower bounds for various graph problems. 
However, it also opens up several interesting directions for future research.

\paragraph{Parallel Repetition}

The \emph{label cover problem} is a special type of a CSP where (i) every constraint is binary and has the so-called projection property (See~\Cref{sec:pcp} for details), and (ii) the underlying graph formed by the constraints is bipartite. 
For $1 \geq c \geq s \geq 0$, let $\mathsf{LabelCover}_{c,s}$ denote the problem that, given a $c$-satisfiable label cover instance, the goal is to compute an $s$-satisfying assignment.
\emph{Parallel Repetition.}~\cite{dinur2014analytical,raz1995parallel} is a powerful framework that reduces $\mathsf{LabelCover}_{1, 1-\varepsilon}$ to $\mathsf{LabelCover}_{1,\varepsilon}$ without increasing the arity of constraints, though it does increase the alphabet size.
The reduction itself is straightforward: for an integer parameter $t \geq 1$, for every choice of $t$ constraints, which are over $\Sigma$, we add a new constraint over $\Sigma^t$ that represents the conjunction of these $t$ constraints. 
As a result, the number of constraints increases from $m$ to $m^t$.
This reduction plays a crucial role in deriving tight inapproximability results for problems like the set cover problem~\cite{feige1998threshold,moshkovitz2012projection} and systems of linear equations~\cite{haastad2001some}.

Unfortunately, it is not clear whether the parallel repetition framework can be used to derive sensitivity lower bounds for $\mathsf{LabelCover}_{1,\varepsilon}$ from those for $\mathsf{LabelCover}_{1,1-\varepsilon}$. 
The challenge arises because a modification to a label cover instance may lead to approximately $m^t - (m-1)^t \approx tm^{t-1}$ modifications in the instance produced by parallel repetition. 
Therefore, a naive lower bound for the latter problem would be $\frac{1}{tm^{t-1}}$ times that of the former problem, which becomes vacuous.
An intriguing open question is whether recent techniques~\cite{bafna2024quasi} can be employed to establish meaningful sensitivity lower bounds for $\mathsf{LabelCover}_{1,\varepsilon}$.

\paragraph{Stronger Sensitivity Lower Bounds.}
The reason that the lower bounds we obtain take the form $\Omega(n^\delta)$ is that, through the reductions, the instance size grows polynomially. 
It is known that there is a PCP of length $O(n \log^2 n)$~\cite{ben2004robust,dinur2007pcp}, and an interesting question is whether we can use the PCP construction to obtain higher lower bounds.
However, note that the lower bound cannot be of the form $n^{1-o(1)}$, because that would imply a lower bound of $n^{1-o(1)}$ for an $n^{-\varepsilon}$-approximation algorithm for the maximum clique problem, which contradicts the fact that there is a trivial $n^{-\varepsilon}$-approximation algorithm with sensitivity $O(n^{1-\varepsilon})$.

\paragraph{Serial Repetition.} 
By applying \emph{serial repetition} to the label cover instance obtained from our PCP theorem --- by taking the ANDs of subsets of constraints ---- we are able to obtain polynomial sensitivity lower bounds against randomized algorithms which output a non-negligible approximation with high probability.
We then use this lower bound to prove \Cref{thm:max-clique-intro}. 
When we are instead interested in polynomial-time tractability, rather than low sensitivity, there is not much of a difference between such with-high-probability guarantees and guarantees in expectation. 
Indeed, we can simply compute multiple assignments and take the best one that we find. However, this is not a sensitivity-preserving process. 
Indeed, it is not clear whether with-high-probability guarantees and in-expectation guarantees are equivalent for low-sensitivity algorithms, making it an interesting open problem to extend \Cref{thm:max-clique-intro} to approximation algorithms with in-expectation guarantees.

\paragraph{Hardness for Polynomial-time Tractable Problems.}
Although our lower bound argument begins with $\mathsf{E2LIN}$, a polynomial-time tractable CSP, the CSP we obtain at the end of the PCP construction is NP-hard (as can be verified using Schaefer’s dichotomy theorem~\cite{schaefer1978complexity}). 
Therefore, with the current approach, we can only establish lower bounds for NP-hard problems. 
An interesting open question is whether similar lower bounds can be established for polynomial-time tractable problems, such as finding a $(1-\varepsilon)$-satisfying assignment for satisfiable $\mathsf{E3LIN}$ instances, where each constraint is of the form $x + y + z = 0 \pmod 2$ or $x + y + z = 1 \pmod 2$.

Another potential approach is to use reductions from SAT to polynomial-time tractable problems studied in the area of fine-grained complexity~\cite{williams2018some}, though one must carefully analyze the constants $C_I$ and $C_\sigma$.
Since such reductions typically blow up the instance size exponentially, the resulting sensitivity lower bound will be logarithmic at best.

\paragraph{Average Sensitivity.}
Average sensitivity~\cite{murai2019sensitivity,varma2023average} is a variation of sensitivity where we measure an algorithm's stability against average-case modifications to the instance. 
Specifically, the average sensitivity of an algorithm $A$ on a graph $G=(V,E)$ is defined as the average Hamming distance between $A(G)$ and $A(G - e)$, where the average is over edges $e \in E$ deleted from $G$.
We note that the sensitivity of an algorithm provides an upper bound on average sensitivity, making the latter easier to bound. 
Algorithms with low average sensitivity have been proposed for various problems, including graph problems~\cite{varma2023average}, 
dynamic programming problems~\cite{kumabe2022average-knapsack,kumabe2022average-dp}, clustering problems~\cite{peng2020average,yoshida2022average}, and learning problems~\cite{hara23average}.

An interesting question is whether the techniques developed in this work can be applied to bound average sensitivity. 
Our approach heavily relies on the fact that, in the context of CSPs, the sensitivity with respect to deleting a constraint can be lower bounded by half of the sensitivity with respect to swapping a constraint. 
However, this relationship does not generally hold for average sensitivity, and a more delicate argument is needed to account for changes in the underlying graph of a CSP instance.

\paragraph{Hardness for Distributed Models.}
In the $\mathsf{LOCAL}$ model for distributed algorithms on graphs, locality (or round complexity) lower bounds are typically proved by constructing two graphs $G_1$ and $G_2$ that are locally indistinguishable even though the optimal values on $G_1$ and $G_2$ are significantly different~\cite{coupette2021breezing,kuhn2004cannot}.
Indeed, some of these arguments can be used to obtain lower bounds for the non-signaling model~\cite{gavoille2009can,balliu2025new,coiteux2024no}.\footnote{They showed this fact for fractional problems, but it holds for constant-factor approximation to the minimum vertex cover problem because the problem has an integrality gap of two, which is constant.}
In particular, Balliu~et~al.~\cite{balliu2025new} showed that any non-signaling distribution for $(\log^{O(1)} \min\{\Delta,n\})$-approximation of the minimum vertex cover problem on graphs with maximum degree $\Delta$ requires locality $\Omega(\min\{\log \Delta /\log \log \Delta,\sqrt{\log n /\log \log n}\})$.
By contrast, we show that $(1+\varepsilon)$-approximation requires locality $\Omega(\log n)$ even when $\Delta$ is constant.

We note that several other extensions of the $\mathsf{LOCAL}$ model have been studied, including the sequential-$\mathsf{LOCAL}$ ($\mathsf{SLOCAL}$)~\cite{ghaffari2017complexity}, dynamic-$\mathsf{LOCAL}$~\cite{etessami2023locality}, and online-$\mathsf{LOCAL}$ models~\cite{etessami2023locality}. 
In these settings, vertices (or sometimes edges, depending on the model) are presented in an adversarial order, say, $v_1,\ldots,v_n$, and the algorithm must decide or update labels accordingly.

In $\mathsf{SLOCAL}$ with locality $t$, the label of $v_i$ is determined from the topology of its $t$-hop neighborhood together with the labels already assigned to the vertices $v_1,\ldots,v_{i-1}$ inside that neighborhood.
In dynamic-$\mathsf{LOCAL}$ with locality $t$, we begin with an empty graph and incrementally add vertices $v_i$ along with edges connecting them to $v_1,\ldots,v_{i-1}$. 
After each insertion, the algorithm may update the label of $v_i$ and of any vertices within its current $t$-hop neighborhood, using the local topology and the labels available at that moment.
In online-$\mathsf{LOCAL}$ with locality $t$, we likewise start from an empty graph and add vertices $v_i$ along with edges connecting them to $v_1,\ldots,v_{i-1}$. 
Here, however, the algorithm must irrevocably assign a label to $v_i$ based solely on the topology and labels in its current $t$-hop neighborhood at that moment.

Whether our techniques yield lower bounds for these models remains an open question. 
To see the difficulty, suppose we wish to construct a low-sensitivity algorithm $A$ from an algorithm $A'$ in any of the above-mentioned models.
Fix a vertex ordering $v_1,\ldots,v_n$, simulate $A'$ on that order, and output the final labels produced by $A'$ as the output of $A$.
This $A$ is indeed low-sensitivity with respect to deleting the last vertex $v_n$; however, we gain no guarantees about its behavior if some vertex in the middle of the ordering is deleted. Indeed, the sequential nature of these models means that the effect of deleting this middle vertex could change the labels of many of the vertices of the graph, propagated through the update steps of latter vertices in the order. 

\subsection{Related Work}\label{subsec:related-work}

Differential privacy~\cite{dwork2006differential} is a fundamental concept in private data analysis, requiring that the output distributions of an algorithm be similar for neighboring instances. 
Though we do not define it formally here, it is straightforward to show that an $\varepsilon$-differentially private algorithm implies an algorithm with sensitivity $O(\varepsilon n)$, where $n$ is a bound on the output size. 
By combining this with known differentially private algorithms, we can derive approximation algorithms with sensitivity $O(\varepsilon n)$ for the global minimum cut problem~\cite{gupta2010differentially}, the densest subgraph problem~\cite{dhulipala2022differential}, and the correlation clustering problem~\cite{cohen2022near}, although this sensitivity bound is relatively weak.

Aside from the linear sensitivity lower bound for $2$-coloring, the only other known lower bound we are aware of is that any \emph{deterministic} constant-factor approximation algorithm for the maximum matching problem requires $\Omega(\log^* n)$ queries~\cite{yoshida2021sensitivity}, where $\log^* n$ is the iterated logarithm of $n$. 
However, their argument relies on Ramsey theory, and thus it cannot be extended to randomized approximation algorithms.

Individual fairness \cite{DworkHPRZ12} is another area which shares similarities with sensitivity. Individual fairness requires that similar individuals should be assigned similar labels (more specifically, their distributions over labels should have low statistical distance). Hence, our sensitivity bounds give lower bounds in the case when individuals are represented as graphs and are classified according to, for example, their cliques or cuts; such representations have been used, for example in \cite{KangHMT20}.

\section{Preliminaries}\label{sec:pre}

For positive integer $n$, let $[n]$ denote the set $\{1,2,\ldots,n\}$.
We use bold symbols to denote random variables.
For a real number $x \in \mathbb R$, let $[x]_+ = \max\{x,0\}$.

\paragraph{Graphs.}
We often use $n$ and $m$ to denote the number of vertices and edges in a graph when the graph is clear from context.
For a graph $G=(V,E)$ and a vertex $v\in V$, let $\deg(v)$ denote the degree of $v$.
For a set $E$ and an element $e \in E$, let $E-e$ denote the set $E \setminus \{e\}$.

For a graph $G=(V,E)$, let $\lambda_i(G)$ denote the $i$-th largest eigenvalue of the adjacency matrix of $G$.
It is known that $\lambda_1=d$ when $G$ is $d$-regular.
Let $\lambda(G) = \max\{\lambda_2(G), |\lambda_n(G)|\}$.
A $d$-regular graph $G=(V,E)$ is called an \emph{$(n, d, \lambda)$-expander} if $ \lambda(G) \leq \lambda < d$. The following lemmas guarantee the existence of sufficiently strong expanders.

\begin{lemma}[see, e.g., \cite{hoory2006expander}]\label{lem:expander-construction}
    Let $d \geq 3$ be an integer.
    Then, there exist explicit constant $\lambda \leq d/2$ and an explicit (polynomial-time computable) family of $(n, d, \lambda)$-expanders.    
\end{lemma}

\begin{lemma}[see, e.g., \cite{hoory2006expander}]\label{lem:expander-superimpose}
    If $G = (V,E)$ is a $d$-regular graph and $G' = (V,E')$ is a $d'$-regular graph, then $H =(V,E \cup E')$ is a $(d+d')$-regular graph such that $\lambda(H) \leq \lambda(G) + \lambda(G')$.
\end{lemma}

\paragraph{Constraint satisfaction problems.}

We formally define \emph{constraint satisfaction problems} (CSPs).
An instance of a CSP is a tuple $I = (V,E,\Sigma,\mathcal{R} = \{R_e\}_{e\in E})$ consisting of a variable set $V$, a set of hyperedges $E$ over $V$, a finite domain $\Sigma$ (also referred to as the alphabet), and a relation $R_e  \subseteq \Sigma^{|e|}$ for each $e \in E$.
A \emph{constraint} refers to a pair $(e,R_e)$ for $e \in E$.
We say that an assignment $\sigma:V\to \Sigma$ \emph{satisfies} a constraint $(e = (v_1,\ldots,v_k),R_e)$ if $(\sigma(v_1),\ldots,\sigma(v_k)) \in R_e$.
Also, we say that $\sigma$ satisfies $I$ if it satisfies all the constraints.
The goal of a CSP is, given an instance $I$ of the CSP, to find a satisfying assignment for $I$.

Now, we consider an optimization version of CSPs.
For a CSP instance $I=(V,E,\Sigma,\mathcal{R})$ and an assignment $\sigma:V\to \Sigma$, let $\val_I(\sigma)$ denote the fraction of constraints in $I$ satisfied by $\sigma$.
We define $\cost_I(\sigma) = 1 - \val_I(\sigma)$ as the fraction of constraints violated by $\sigma$.
Let $\opt(I) = \max_{\sigma:V \to \Sigma}\val_I(\sigma)$.
For $1 \geq c \geq s \geq 0$, we define $\mathsf{MaxCSP}_{c,s}$ as the problem that, given an instance $I=(V,E,\Sigma,\mathcal{R})$ with $\opt(I) \geq c$, the goal is to find an assignment $\sigma:V \to \Sigma$ with $\val_I(\sigma) \geq s$.\footnote{Usually, $\mathsf{MaxCSP}_{c,s}$ refers to the decision problem where the goal is to distinguish instances $I$ with $\opt(I) \geq c$ from instances $I$ with $\opt(I) \leq s$. 
However, we define it as a search problem because we are focused on the sensitivity of algorithms.}

For two instances $I = (V,E,\Sigma,\{R_e\}_{e \in E})$ and $\tilde I = (V,E,\Sigma,\{\tilde R_e\}_{e \in E})$ on the same underlying hypergraph and domain, we define the \emph{swap distance} between $I$ and $\tilde I$ as $\swapdist(I,\tilde I) := \#\{R_e \neq \tilde R_e : e \in E\}$.

\paragraph{Sensitivity.}

Let $I=(V,E,\Sigma,\{R_e\}_{e \in E})$ be a CSP instance.
For a hyperedge $e \in E$, let $I-e$ denote the instance $(V,E-e,\Sigma,\{R_f\}_{f \in E - e})$.
For two assignments $\sigma,\sigma':V \to \Sigma$, let $\mathrm{Ham}(\sigma,\sigma') = \#\{v \in V: \sigma(v) \neq \sigma(v')\}$ denote their Hamming distance.
The \emph{sensitivity} of a deterministic algorithm $A$ on $I$ is defined as
\begin{align}
    \sens(A, I) ~:=~ \max_{e \in E}\mathrm{Ham}(A(I), A(I - e)),
    \label{eq:deterministic-sensitivity}    
\end{align}
where $A(I)$ denotes the output assignment of $A$ on $I$.

For a distribution $\mu$ over $X$ and another distribution $\tilde \mu$ over $\tilde X$, we say that a joint distribution $\pi$ over $X \times \tilde X$ is a \emph{coupling} between them if the marginal distributions on the first and the second coordinates are equal to $\mu$ and $\tilde \mu$, respectively.
Let $\Pi(\mu,\tilde \mu)$ denote the set of all couplings between $\mu$ and $\tilde \mu$.
For two distributions $\mu,\tilde \mu$ over assignments on the same domain, we define the earth mover's distance between them as 
\[
    \EMD(\mu,\tilde \mu) = \min_{\pi \in \Pi(\mu,\tilde \mu)} \E_{(\sigma,\tilde \sigma)\sim \pi} \mathrm{Ham}(\sigma,\tilde \sigma).
\]
The \emph{sensitivity} of a randomized algorithm $A$ on a CSP instance $I=(V,E,\Sigma,\mathcal R)$ is defined as
\[
    \sens(A, I) := \max_{e \in E} \EMD(A(I), A(I - e)),
\]
where we identify random variables $A(I)$ and $A(I-e)$ with their distributions.
Note that this definition matches~\eqref{eq:deterministic-sensitivity} when the algorithm is deterministic.
For a family of algorithms $\mathcal A$ and a family of instances $\mathcal I$ over the same domain $\Sigma$, we define 
\[
    \sens(\mathcal A,\mathcal I) := \min_{A \in \mathcal A} \max_{I \in \mathcal I} \sens(A,I).
\]
Also, we define $\sens_{c,s}(\mathcal I) := \sens(\mathcal A,\mathcal I)$, where $\mathcal A$ is the family of all possible algorithms for $\mathsf{MaxCSP}_{c,s}$.

We define a variant of sensitivity which is more convenient when studying CSPs.
Let $I=(V,E,\Sigma,\{R_e\}_{e\in E})$ be a CSP instance.
For $e \in E$ and $R \subseteq \Sigma^{|e|} $, $I^{e \gets R}$ denote the instance obtained from $I$ by replacing $R_e$ with $R$.
Then, the \emph{swap sensitivity} of a randomized algorithm $A$ on an instance $I = (V,E,\Sigma,\mathcal R)$ is defined as 
\[
    \swapsens(A, I) := \max_{e \in E}\max_{R \subseteq \Sigma^{|e|}} \EMD(A(I), A(I^{e\gets R})).
\]
We define $\swapsens(\mathcal A,\mathcal I)$ and $\swapsens_{c,s}(\mathcal I)$ as with sensitivity.
We say that a family of instances $\mathcal I$ is \emph{swap-closed} if for any $I =(V,E,\Sigma,\mathcal R)\in \mathcal I$, $e \in E$, and $R \subseteq \Sigma^{|e|}$, the instance $I^{e \gets R}$ also belongs to $\mathcal I$.
Note that for any $A\in \mathcal A$ and $I,\tilde I \in \mathcal I$ for a swap-closed family of instances $\mathcal I$, we have $\EMD(A(I), A(\tilde I)) \leq \swapsens(\mathcal A,\mathcal I) \cdot \swapdist(I,\tilde I)$.
For a family of instances $\mathcal I$, we define $\swapclo(\mathcal I)$ as the \emph{swap-closure} of $\mathcal I$, i.e., the family of instances obtained by (repeatedly) swapping constraints in $I \in \mathcal I$.

The following lemma shows that we can reduce the problem of bounding sensitivity to that of bounding swap sensitivity.
\begin{lemma}\label{lem:sens-and-swap-sens}
    For any family of algorithms $\mathcal A$ and family of CSP instances $\mathcal I$, we have
    \[
        \sens(\mathcal A,\mathcal I) \geq \frac{1}{2}\swapsens(\mathcal A,\mathcal I).
    \]
\end{lemma}
\begin{proof}
    Let $A \in \mathcal A$ be an algorithm that attains $\sens(\mathcal A,\mathcal I)$.
    For any $I = (V,E,\Sigma,\{R_e\}_{e \in E}) \in \mathcal I$, $e \in E$, and $R \subseteq \Sigma^{|e|}$, we have
    \begin{align*}
        \EMD(A(I),A(I^{e \gets R}))
        \leq 
        \EMD(A(I),A(I-e))
        +
        \EMD(A(I-e), A(I^{e \gets R}))
        \leq 2\sens(\mathcal A,\mathcal I).
    \end{align*}
    Hence, we have 
    \[
        \swapsens(\mathcal A,\mathcal I)
        \leq 
        \max_{I \in \mathcal I} \swapsens(A, I)   
        \leq 
        \max_{I \in \mathcal I, e \in E, R \subseteq \Sigma^{|e|}} \EMD(A(I),A(I^{e \gets R}))
        \leq         
        2\sens(\mathcal A,\mathcal I).
        \qedhere
    \]
\end{proof}

\section{Sensitivity-Preserving Reductions}\label{sec:reduction}

The PCP framework can be viewed as a sequence of reductions between CSPs. 
Therefore, it is useful to introduce a template for the reductions that we will use throughout our analysis.
\begin{definition}\label{def:sensitivity-reduction}
    Let $\mathcal I$ be a family of CSP instances on the same underlying hypergraph and domain.
    Let $T_I$ be a procedure that transforms a CSP instance $I \in \mathcal I$ to another CSP instance $I'$, and let $T_\sigma$ be a (possibly randomized) procedure that transforms an assignment $\sigma':V' \to \Sigma'$ for $I'$ to an assignment $\bm \sigma:V \to \Sigma$ for $I$ ($T_\sigma$ might depend on $I$).
    For $c,s,c',s' \in [0,1]$ and $C_I,C_\sigma>0$, we say that the pair $(T_I,T_\sigma)$ is a $(c,s,c',s',C_I,C_\sigma)$-sensitivity-preserving reduction for $\mathcal I$ if the following holds:
    \begin{enumerate}
        \item If $\mathsf{opt}(I)\geq c$, then $\mathsf{opt}(I') \geq c'$. \label{item:reduction-completeness}
        \item If a (possibly random) assignment $\bm \sigma'$ for $I'$ satisfies $\E[\mathsf{val}_{I'}(\bm \sigma')] \geq s'$, then the assignment $\bm \sigma = T_\sigma(\bm \sigma')$ satisfies $\E[\mathsf{val}_{I}(\bm \sigma)] \geq s$. \label{item:reduction-soundness}
        \item Let $I,\tilde I \in \mathcal I$ be two CSP instances.
        Then, we have $\mathsf{SwapDist}(T_I(I),T_I(\tilde I)) \leq C_I \cdot \mathsf{SwapDist}(I,\tilde I)$. 
        In particular, this implies that $T_I$ generates CSP instances on the same underlying hypergraph and domain for any $I \in \mathcal I$.
        \label{item:reduction-instance-distance}
        \item Let $I,\tilde I \in \mathcal I$ be two CSP instances and let $\bm \sigma',\tilde {\bm \sigma}'$ be assignments for $T_I(I)$ and $T_I(\tilde I)$, respectively.
        Then, we have $\EMD(T_\sigma(\bm \sigma'),T_\sigma(\tilde{\bm \sigma}')) \leq C_{\sigma}\cdot  \EMD(\bm \sigma',\tilde{\bm \sigma}')$.
        \label{item:reduction-assignment-distance}
    \end{enumerate}
\end{definition}

\begin{lemma}\label{lem:sensitivity-reduction}
    Let $\mathcal I$ be a swap-closed family of CSP instances on the same underlying hypergraph and the same domain.
    Suppose that there exists a $(c,s,c',s',C_I,C_\sigma)$-sensitivity-preserving reduction $(T_I,T_\sigma)$ for $\mathcal I$.
    Then we have
    \[
        \swapsens_{c',s'}(\swapclo(\mathcal I')) \geq \frac{1}{C_I C_\sigma}\swapsens_{c,s}(\mathcal I),
    \]
    where $\mathcal I' = T_I(\mathcal I) := \{T_I(I) : I \in \mathcal I\}$.
\end{lemma}
\begin{proof}
    Let $A'$ be an algorithm that attains $\swapsens_{c',s'}(\swapclo(\mathcal I'))$.
    Then, we design an algorithm $A$ for $\mathcal I$ using $A'$ as follows:
    Given an instance $I = (V, E, \Sigma, \mathcal{R}) \in \mathcal I$, we construct an instance $I' = (V',E', \Sigma', \mathcal{R}') = T_I(I)$.
    Then, we run the algorithm $A'$ on $I'$ to obtain a (possibly random) assignment $\bm \sigma': V' \to \Sigma'$ for $I'$.
    Then, we output an assignment $\bm \sigma = T_\sigma(\bm \sigma')$ for $\mathcal I$.

    We analyze the approximation guarantee of $A$.
    Suppose $\mathsf{opt}_I(I) \geq c$.
    Then by \Cref{item:reduction-completeness} of \Cref{def:sensitivity-reduction}, we have $\mathsf{opt}(I') \geq c'$.
    Hence, the output assignment $\bm \sigma'$ satisfies $\E\mathsf{val}(I',\bm \sigma') \geq s'$.
    By \Cref{item:reduction-soundness} of \Cref{def:sensitivity-reduction}, we have that $\E\mathsf{val}(I,\bm \sigma) \geq s$.

    Now, we analyze the swap sensitivity of $A$.
    Let $I,\tilde I \in \mathcal I$ be two CSP instances with swap distance one.
    Then by \Cref{item:reduction-instance-distance} of \Cref{def:sensitivity-reduction}, we have $\swapdist(I',\tilde I') \leq C_I \cdot \swapdist(I,\tilde I)$, which implies that \[\EMD(\bm \sigma',\tilde{\bm \sigma}') \leq C_I \cdot \swapsens(A', \swapclo(\mathcal I')) = C_I \cdot  \swapsens_{c',s'}(\swapclo(\mathcal I')).\]
    By \Cref{item:reduction-assignment-distance} of \Cref{def:sensitivity-reduction}, we have $\EMD(T_\sigma(\bm \sigma'),T_\sigma(\tilde{\bm \sigma}')) \leq C_I  C_\sigma \cdot \swapsens_{c',s'}(\swapclo(\mathcal I'))$.
    Hence, we must have $\swapsens_{c',s'}(\swapclo(\mathcal I')) \geq \swapsens_{c,s}(\mathcal I)/(C_I C_\sigma)$.
\end{proof}

\section{Lower Bounds for \texorpdfstring{$\mathsf{E2LIN}$}{E2LIN}}\label{sec:2-xor}

Let $R_0,R_1 \in \mathbb Z_2^2$ be binary relations over $\mathbb Z_2$ such that $(a,b) \in R_0$ if and only if $a + b = 0 \pmod 2$ and $(a,b) \in  R_1$ if and only $a + b = 1 \pmod 2$, respectively.
Then, we define $\mathsf{E2LIN}$ as the CSP over $\mathbb Z_2$ where only $R_0$ and $R_1$ are used to define the constraints.
Also, let $\mathsf{E2LIN}_{c,s}$ be the special case of $\mathsf{MaxCSP}_{c,s}$, where the instances are restricted to those of $\mathsf{E2LIN}$.
In this section, we establish a sensitivity lower bound for $\mathsf{E2LIN}_{1,1-1/2n}$, which we will later amplify using the PCP framework.

Let $\mathcal{I}_n$ denote the set of all $\mathsf{E2LIN}$ instances $(V,E,\mathbb Z_2,\mathcal R)$ such that the graph $(V,E)$ forms a cycle on $n$ vertices.
Note that $\mathcal I_n$ is swap-closed.
Let $\mathcal{A}$ denote the set of randomized algorithms such that the expected number of errors on satisfiable instances in $\mathcal{I}_n$ is at most half. 
We show the following lower bound.
\begin{lemma}\label{lem:2-xor}
    For any even integer $n \geq 4$, we have $\swapsens(\mathcal{A},\mathcal{I}_n) = \Omega(n)$.
    In particular, we have $\swapsens_{1,1-1/2n}(\mathcal I_n) = \Omega(n)$.
\end{lemma}
\begin{proof}
    Let $A \in \mathcal{A}$ be the algorithm that attains $\swapsens(\mathcal{A},\mathcal{I}_n)$.
    Consider an instance $I = (V=(v_1,\ldots,v_n),E, \mathbb Z_2, \{R_e\}_{e \in E}) \in \mathcal{I}_n$ such that $R_e = R_1$ for every $e \in E$.
    Note that $I$ is satisfiable and any assignment violates an even number of constraints in $I$, and hence $A$ must output a satisfying assignment on $I$ with probability at least $3/4$ (otherwise, the expected number of errors is more than $1/4 \cdot 2 = 1/2$).
    Note that there are only two satisfying assignments $\sigma_1,\sigma_2:V\to \mathbb Z_2$ for $I$.
    Without loss of generality, we can assume that they are of the following form:
    \[
        \sigma_1(v_i) = \begin{cases}
            0 & i\text{ is odd}, \\
            1 & i\text{ is even},
        \end{cases}
        \quad
        \sigma_2(v_i) = \begin{cases}
            1 & i\text{ is odd}, \\
            0 & i\text{ is even}.
        \end{cases}
    \]

    Let $e = (v_{n/2},v_{n/2}+1)$ and $e' = (v_n,v_1)$ and consider an instance $\tilde I = I^{e \gets R_0,e' \gets R_0}$.
    As $\tilde I$ is also satisfiable, $A$ must output a satisfying assignment on $I'$ with probability at least $3/4$.
    Note that there are only two satisfying assignments $\tilde \sigma_1, \tilde \sigma_2:V\to \{0,1\}$ for $\tilde I$.
    Without loss of generality, we can assume that they are of the following form:
    \[
        \tilde \sigma_1(v_i) = \begin{cases}
            0 & i\text{ is odd and }i \leq n/2, \\
            1 & i\text{ is even and }i \leq n/2, \\
            1 & i\text{ is odd and }i \geq n/2+1, \\
            0 & i\text{ is even and }i \geq n/2+1,
        \end{cases}
        \quad
        \tilde \sigma_2(v_i) = \begin{cases}
            1 & i\text{ is odd and }i \leq n/2, \\
            0 & i\text{ is even and }i \leq n/2, \\
            0 & i\text{ is odd and }i \geq n/2+1, \\
            1 & i\text{ is even and }i \geq n/2+1.
        \end{cases}
    \]
    For every $i,j \in \{1,2\}$, we have $\Ham(\sigma_i,\tilde \sigma_j ) \geq n/2$.
    Hence, we have
    \[
        \swapsens(\mathcal{A},\mathcal{I}_n)
        \geq 
        \frac{\EMD(A(I),A(\tilde I))}{\swapdist(I,\tilde I)}
        = 
        \frac{1}{2} \cdot \left(1 - \frac{1}{4} - \frac{1}{4}\right) \cdot \frac{n}{2} = \Omega(n).
        \qedhere
    \]
\end{proof}

\section{PCPs and Sensitivity}\label{sec:pcp}

In this section, we show a sensitivity lower bound for a problem called  $\mathsf{LabelCover}$, which is a special case of CSPs.
\begin{definition}
An instance of $\mathsf{LabelCover}$ is a tuple $I = (U,V,E,\Sigma_U,\Sigma_V,\mathcal{R} = \{R_e\}_{e \in E})$, where $(U \cup V, E)$ forms a bipartite graph, $\Sigma_U,\Sigma_V$ are finite domains, and each relation $R_e \subseteq \Sigma_L \times \Sigma_R$ is a projection.
Here, a \emph{projection} refers to a relation of the form $R_e =\{(a,\phi_e(a)) : a \in \Sigma_U\}$ for some map $\phi_e:\Sigma_U \to \Sigma_V$.
For $1 \geq c \geq s \geq 0$, we define $\mathsf{LabelCover}_{c,s}$ as the problem that, given a label cover instance $I$ with $\opt(I) \geq c$, the goal is to find an assignment $\sigma$ for $I$ with $\val_I(\sigma) \geq s$.
\end{definition}
The goal of this section is to show the following:
\begin{theorem}\label{thm:label-cover}
    There exist universal constants $\varepsilon,\delta>0$ and $d,k \geq 1$ such that any algorithm for $\mathsf{LabelCover}_{1,1-\varepsilon}$ on a bipartite graph of maximum degree $d$ and a domain of size $k$ has sensitivity $\Omega(n^\delta)$.
\end{theorem}

As we discussed in \Cref{subsec:proof-overview}, the proof of Theorem~\ref{thm:label-cover} consists of four steps: Degree reduction, expanderization, gap amplification, and alphabet reduction.
We discuss these four steps in \Cref{subsec:degree-reduction,subsec:expanderization,subsec:gap-amplification,subsec:alphabet-reduction}, respectively.
Finally, we prove \Cref{thm:label-cover} in \Cref{subsec:label-cover}.

\subsection{Degree Reduction}\label{subsec:degree-reduction}

In this section, we introduce a transformation that reduces degrees of vertices in the underlying graph of a CSP instance.
Recall that, by \Cref{lem:expander-construction}, there exist universal constants $\lambda_0 < d_0$ such that $(n, d_0, \lambda_0)$-expanders can be explicitly constructed in polynomial time.
Let $I= (V,E,\Sigma,\mathcal R)$ be a $d$-regular CSP instance with $m=|E|$.
We consider the following procedure, called \Call{DegreeReduction}{}, to construct an instance $I' = (V',E',\Sigma',\mathcal R')$:
\begin{itemize}
    \item Replace each vertex $v \in V$ by $d$ many vertices to get the new vertex set $V'$. 
    Denote the set of new vertices corresponding to $v$ by $\mathrm{cloud}(v)$. 
    Each vertex in $\mathrm{cloud}(v)$ naturally corresponds with a neighbor of $v$ from $G=(V,E)$.
    \item For each edge $e \in E$, place an ``inter-cloud'' edge $e'$ in $E'$ between the associated cloud vertices. 
    This gives exactly one inter-cloud edge per vertex in $V'$.
    Whatever the old constraint $R_e$ on $e$ was, put the exact same constraint on $e'$.
    \item For each $v \in V$, put a $(d, d_0, \lambda_0)$-expander on $\mathrm{cloud}(v)$ given by \Cref{lem:expander-construction}. 
    Further, put equality constraints on these expander edges.
\end{itemize}
We can observe that in this process each new vertex in $V'$ has degree exactly equal to $d_0 + 1$.
Thus we have created a $(d_0 + 1)$-regular graph, as desired. 
Also the number of newly added edges is equal to $\sum_{u\in V} d d_0/2 = d_0 \sum_{u\in V} d/2 = d_0m$, and hence $m'=(d_0 +1)m$.
Note that the domain $\Sigma$ does not change and in particular $\Sigma'=\Sigma$. A depiction is given in \Cref{fig:degreeRed}.

\begin{figure}[h]
    \centering
    \begin{subfigure}[b]{0.4\textwidth}
    \centering
    \begin{tikzpicture}[scale=0.9]
        \node[circle, draw=black!80, fill=green!9, minimum size=1pt, thick] at (0,0)(v1){\small \phantom{a}};

        \node[circle, draw=black!80, fill=red!12, minimum size=1pt, thick] at (1,1.4)(n1){\small \phantom{a}};

        \node[circle, draw=black!80, fill=red!12, minimum size=1pt, thick] at (0,-1.5)(n2){\small \phantom{a}};

        \node[circle, draw=black!80, fill=red!12, minimum size=1pt, thick] at (1.5,-0.5)(n3){\small \phantom{a}};

        \node[circle, draw=black!80, fill=red!12, minimum size=1pt, thick] at (-1.5,0.5)(n4){\small \phantom{a}};

        \draw[very thick, draw=black!80] (v1) --(n1);

        \draw[very thick, draw=black!80] (v1) --(n2);

        \draw[very thick, draw=black!80] (v1) --(n3);

        \draw[very thick, draw=black!80] (v1) --(n4);

        \node at (-.2,0.5)(tau2){$v$};
    \end{tikzpicture}
    \caption{Vertex $v$ prior to \Call{DegreeReduction}{}.}
    \end{subfigure}
    \quad
    \begin{subfigure}[b]{0.4\textwidth}
    \centering
    \begin{tikzpicture}[scale=0.9]
        \filldraw[color=black!0,rounded corners=10mm, fill=green!9, very thick] (-0.8,-2.2) -- (2.2,-1.3) -- (1.8,2) --(-2,1) --cycle;

        \node[circle, draw=black!80, fill=betterYellow!16, minimum size=1pt, thick] at (1.2,1.3)(v1){\small \phantom{a}};

        \node[circle, draw=black!80, fill=betterYellow!16, minimum size=1pt, thick] at (-0.3,-1.4)(v2){\small \phantom{a}};

        \node[circle, draw=black!80, fill=red!12, minimum size=1pt, thick] at (2,2.4)(n1){\small \phantom{a}};

        \node[circle, draw=black!80, fill=red!12, minimum size=1pt, thick] at (0,-2.5)(n2){\small \phantom{a}};

        \node[circle, draw=black!80, fill=betterYellow!16, minimum size=1pt, thick] at (1.5,-0.8)(v3){\small \phantom{a}};

        \node[circle, draw=black!80, fill=red!12, minimum size=1pt, thick] at (2.5,-1.5)(n3){\small \phantom{a}};

        \node[circle, draw=black!80, fill=betterYellow!16, minimum size=1pt, thick] at (-1.,0.5)(v4){\small \phantom{a}};

        \node[circle, draw=black!80, fill=red!12, minimum size=1pt, thick] at (-2.5,1.5)(n4){\small \phantom{a}};

        \draw[very thick, draw=black!80] (v1) --(n1);

        \draw[very thick, draw=black!80] (v2) --(n2);

        \draw[very thick, draw=black!80] (v3) --(n3);

        \draw[very thick, draw=black!80] (v4) --(n4);

        \draw[very thick, draw=black!80] (v4) --(v1);

        \draw[very thick, draw=black!80] (v1) --(v2);

        \draw[very thick, draw=black!80] (v2) --(v3);

        \draw[very thick, draw=black!80] (v4) --(v3);

        \node at (-.3,1.8)(tau2){$\mathsf{cloud}(v)$};

    \end{tikzpicture}
    \caption{$\mathsf{cloud}(v)$ after \Call{DegreeReduction}{}.}
    \end{subfigure}
    \caption{\Call{DegreeReduction}{} on a vertex $v$ with $d=4$. The intra-cloud edges represent an expander with $d_0=2$.}
    \label{fig:degreeRed}
\end{figure}
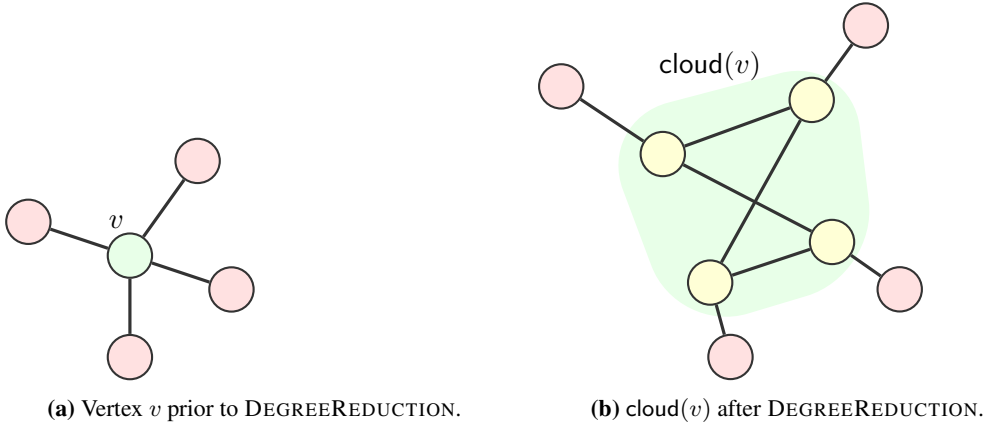

We now show how the sensitivity bound translates through the transformation.
\begin{lemma}\label{lem:degree-reduction}
    Let $\mathcal I$ be a swap-closed family of binary $d$-regular CSP instances over $\Sigma$ and let $\mathcal I' = \Call{DegreeReduction}{\mathcal I}$.
    Then for any $\varepsilon>0$, we have  
    \[
        \swapsens_{1,1-\varepsilon'}(\swapclo(\mathcal I')) \geq d \cdot \swapsens_{1,1-\varepsilon}(\mathcal I),
    \] 
    for $\varepsilon' := \varepsilon/C$, where $C>0$ is a universal constant.
\end{lemma}
We note that this process indeed \emph{increases} the sensitivity lower bound, which is crucial for offsetting the sensitivity decrease in the gap amplification and alphabet reduction steps.
\begin{proof}[Proof of \Cref{lem:degree-reduction}]
    Let $I = (V, E, \Sigma, \mathcal R) \in \mathcal I$ and $I' = (V',E', \Sigma,\mathcal R') = \Call{DegreeReduction}{\mathcal I}$.
    We consider an algorithm $T_\sigma$ that, given an assignment $\sigma': V' \to \Sigma$ for $I'$, constructs an assignment $\bm \sigma:V\to \Sigma$ for $I$ by setting $\bm \sigma(u) = \sigma'(\bm u')$, where $\bm u' \in V'$ is a vertex sampled from $\mathrm{cloud}(u)$ uniformly at random.

    We show that the pair $(\Call{DegreeReduction}{},T_\sigma)$ is a $(1,1-\varepsilon,c'=1,s' = 1-\varepsilon',C_I=1,C_\sigma = 1/d)$-sensitivity-preserving reduction, from which the claim follows.
    As the analysis for $c'=1$ and $C_I=1$ is obvious, we analyze $s'$ and $C_\sigma$ below.

    Let $\bm \sigma'$ be an assignment for $I'$ with $\E \val_{I'}(\bm \sigma') \geq 1-\varepsilon' = s'$.
    Our goal is to show that $\E \val_I(\bm \sigma) \geq 1-\varepsilon$.
    We condition on $\bm \sigma' = \sigma'$ and aim to show that $\E[\cost_I(\bm \sigma)] \leq C \cdot \E[\cost_{I'}(\sigma')]$. 
    We then obtain the result by unconditioning  $\bm \sigma'$.

    For a vertex $u \in V$, let us define $\bm S^u$ to be the set of vertices in $\mathrm{cloud}(u)$ on which $\sigma'$ disagrees with $\bm \sigma(u)$. 
    Suppose $e = (u, v) \in E$ is one of the edges in $I$ that are violated by $\bm \sigma$. 
    Let $e'$ be the corresponding inter-cloud edge in $E'$. 
    The key observation is that either $\sigma'$ violates the edge $e'$ or one of the endpoints of $e'$ belongs to $\bm S^u$ or $\bm S^v$. 
    Thus we have
    \[
        \E \cost_{I}(\bm \sigma) \cdot m
        \leq \E[\cost_{I'}(\sigma')m' + \sum_{u \in V}|\bm S^u|]. 
    \]
    It follows that either
    \begin{itemize}
        \item[a)] $\E \cost_{I'}(\sigma') \cdot m' \geq \E \cost_{I}(\bm \sigma) \cdot m/2$, or
        \item[b)] $\sum_{u\in V} \E|\bm S^u| \geq \E \cost_{I}(\bm \sigma) \cdot m/2$.
    \end{itemize}
    In case a), we obtain
    \[
        \cost_{I'}(\sigma ') \cdot m' 
        \geq \frac{\E \cost_{I}(\bm \sigma) \cdot m}{2} = \frac{\E \cost_{I}(\bm \sigma) \cdot m'}{2(d_0 + 1)}.
    \]
    Hence, we have $\E \cost_{I}(\bm \sigma) \leq C \cdot \E \cost_{I'}(\sigma')$ by setting $C \geq 2(d_0+1)$.

    To handle case b), for each label $a \in \Sigma$, let $C^u_a = (\sigma')^{-1}(a) \cap \mathrm{cloud}(u)$ be the set of vertices in $\mathrm{cloud}(u)$ that are labelled $a$ by $\sigma'$ and $p^u_a = \frac{|C^u_a|}{|\mathrm{cloud}(u)|}$ be its fraction.
    Then, note that
    \[
    \E |\bm S^u| = \sum_{a \in \Sigma}p^u_a |\mathrm{cloud}(u) \setminus C^u_a| = |\mathrm{cloud}(u)| \cdot \sum_{a \in \Sigma}p^u_a(1-p^u_a), 
    \]
    where the first equality follows as $p_a^u$ is the probability that we choose label $a$ as our label for $\bm \sigma$, and hence every vertex in $\mathrm{cloud}(u) \setminus C^u_a$ differs from the label given to $\bm \sigma$.
    Note that every edge between $C^u_a$ and $C^u_b$ for $a \neq b$ is violated by $\sigma'$ because they are all labelled with ``equality'' constraints.
    Then by the fact that the cloud is an expander, 
    there exists a constant $\phi>0$ that is determined by $d_0$ and $\lambda_0$ such that the number of edges in $\mathrm{cloud}(u)$ violated by $\sigma'$ is at least
    \begin{align*}
        & \frac{\phi}{2} \sum_{a \in \Sigma} \min\{|C^u_a|, |\mathrm{cloud}(u) \setminus C^u_a|\} 
        = \frac{\phi |\mathrm{cloud}(u)|}{2} \cdot \sum_{a \in \Sigma} \min\{p^u_a, 1-p^u_a\} \\
        & \geq \frac{\phi |\mathrm{cloud}(u)|}{2} \cdot \sum_{a \in \Sigma} p^u_a(1-p^u_a)
        = \frac{\phi}{2} \E |\bm S^u|.
    \end{align*}

    Therefore $\sigma'$ violates at least the following number of edges:
    \begin{align*}
        & \cost_{I'}(\sigma') \cdot m'
        \geq \frac{\phi}{2} \sum_{u \in V} \E|\bm S^u| \\
        & \geq \frac{\phi \E \cost_I(\bm \sigma) \cdot m}{4} \tag{since we are in case (b)} \\
        & = \frac{\phi \E \cost_I(\bm \sigma) \cdot m'}{4(d_0 + 1)}.
    \end{align*}
    Hence, we have $\E \cost_I(\bm \sigma) \leq C \cdot \cost_{I'}(\sigma')$ by setting $C \geq 4(d_0+1)/\phi$.

    Next, we analyze the value of $C_\sigma$.
    Let $I,\tilde I \in \mathcal I$ be CSP instances and let $I'=\Call{DegreeReduction}{I}$ and $\tilde I'=\Call{DegreeReduction}{\tilde I}$.
    Let $\sigma',\tilde \sigma':V'\to \Sigma$ be assignments for $\mathcal I',\tilde{\mathcal I}'$, respectively.
    Let $\pi \in \Pi(\bm \sigma',\tilde{\bm \sigma}')$ such that 
    $\mathrm{EMD}(\bm \sigma',\tilde{\bm \sigma}') = \E_{(\sigma',\tilde \sigma') \sim \pi}\mathrm{Ham}(\sigma',\tilde \sigma')$.
    Then, we have 
    \begin{align*}
        & \EMD(\bm \sigma,\tilde{\bm \sigma}) 
        \leq \E_{(\sigma',\tilde \sigma') \sim \pi} \sum_{u \in V} \TV(\bm \sigma(u) \mid \sigma', \tilde{\bm \sigma}(u) \mid \tilde \sigma')      \\
        & = \E_{(\sigma',\tilde \sigma') \sim \pi}  \sum_{u \in V} \frac{1}{2}\sum_{a \in \Sigma} |\frac{|( \sigma')^{-1}(a) \cap \mathrm{cloud}(u)|}{|\mathrm{cloud}(u)|} - \frac{|(\tilde{\sigma}')^{-1}(a) \cap \mathrm{cloud}(u)|}{|\mathrm{cloud}(u)|}| \\
        & = \frac{1}{d} \E_{(\sigma',\tilde \sigma') \sim \pi} \sum_{u \in V'}\bm{1}[\sigma'(u) \neq \tilde{\sigma}'(u)] \\
        & = \frac{\EMD(\bm \sigma', \tilde{\bm \sigma}')}{d}.
    \end{align*}
    Hence, the choice $C_\sigma = 1/d$ satisfies \Cref{item:reduction-assignment-distance} of \Cref{def:sensitivity-reduction}.
\end{proof}

We will need a variant of \Cref{lem:degree-reduction} in the alphabet reduction step, which we discuss below.
First, we note that \Call{DegreeReduction}{} can be extended to non-regular instances by introducing  $\deg(v)$ copies of each vertex $v \in V$.
Note that the resulting graph is $(d_0 + 1)$-regular.

\paragraph{Marked CSPs.}
A \emph{marked CSP instance} $\hat I=(V,E,\Sigma,\mathcal R,S)$ consists of a CSP instance $(V,E,\Sigma,\mathcal R)$ and a set of \emph{marked vertices} $S\subseteq V$.\\

For a marked CSP instance $\hat I$, we measure the sensitivity of an algorithm $A$ using marked vertices only.
Specifically, we define
\[
    \sens(A,\hat I) = \max_{e \in E}\Ham(A(I)\mid_S, A(I-e)\mid_S),
\]
where for an assignment $\sigma : V \to \Sigma$, $\sigma|_S : S \to  \Sigma$ is a restriction of $\sigma$ on $S$.
Then, we can define other sensitivity-related notions for marked CSP instances using the definition above.

For a marked instance $\hat I$, we define  $\Call{DegreeReduction}{\hat I}$ as a marked instance $(V',E',\Sigma,\mathcal R',S')$, where $(V',E',\Sigma,\mathcal R')= \Call{DegreeReduction}{I}$ and $S' = \bigcup_{v \in S}\mathrm{cloud}(v)$.
We can define the swap-closed property and the swap-closure for a family of marked instances naturally.
The proof of \Cref{lem:degree-reduction} gives the following:
\begin{corollary}\label{cor:degree-reduction}
    Let $\hat{\mathcal I}$ be a swap-closed family of binary marked CSP instances over $\Sigma$, where every marked vertex has degree at least $d$, and let $\hat{\mathcal I}' = \Call{DegreeReduction}{\hat{\mathcal I}}$.
    Then for any $\varepsilon>0$, we have  
    \[
        \swapsens_{1,1-\varepsilon'}(\swapclo(\hat{\mathcal I}')) \geq d \cdot \swapsens_{1,1-\varepsilon}(\hat{\mathcal I}),
    \] 
    for $\varepsilon' := \varepsilon/C$, where $C>0$ is a universal constant.  
\end{corollary}

\subsection{Expanderization}\label{subsec:expanderization}
In this section, we introduce a transformation that makes the underlying graph of a CSP instance into an expander.
This subroutine, called \Call{Expanderization}{}, is simple. 
Given a binary $d$-regular instance $I$ on $n$ variables, we just superimpose an $(n, d_0, \lambda_0)$-expander given by \Cref{lem:expander-construction}. (This may lead to multiple edges.) 
On each edge of the expander we simply put a trivial constraint, i.e., a constraint that is always satisfied. A depiction can be seen in \Cref{fig:expanderize}.

\begin{figure}
    \centering
    \begin{subfigure}[b]{0.4\textwidth}
    \centering
    \begin{tikzpicture}[scale=0.8]

        \node[circle, draw=black!80, fill=green!9, minimum size=1pt, thick] at (1.2,1.3)(v1){\small \phantom{a}};

        \node[circle, draw=black!80, fill=green!9, minimum size=1pt, thick] at (-0.3,-1.4)(v2){\small \phantom{a}};

        \node[circle, draw=black!80, fill=green!9, minimum size=1pt, thick] at (0,2.3)(v5){\small \phantom{a}};

        \node[circle, draw=black!80, fill=green!9, minimum size=1pt, thick] at (0,-2.5)(n2){\small \phantom{a}};

        \node[circle, draw=black!80, fill=green!9, minimum size=1pt, thick] at (1.5,-0.8)(v3){\small \phantom{a}};

        \node[circle, draw=black!80, fill=green!9, minimum size=1pt, thick] at (2.5,-1.5)(n3){\small \phantom{a}};

        \node[circle, draw=black!80, fill=green!9, minimum size=1pt, thick] at (-1.,0.5)(v4){\small \phantom{a}};

        \node[circle, draw=black!80, fill=green!9, minimum size=1pt, thick] at (-2.5,1.5)(n4){\small \phantom{a}};

        \draw[very thick, draw=black!80] (v2) --(n2);

        \draw[very thick, draw=black!80] (v5) --(v4);

        \draw[very thick, draw=black!80] (n2) --(n3);

        \draw[very thick, draw=black!80] (v4) --(n4);

        \draw[very thick, draw=black!80] (v3) --(n3);

        \draw[very thick, draw=black!80] (v2) --(v5);

        \draw[very thick, draw=black!80] (v1) --(v3);

    \end{tikzpicture}
    \caption{$G$ before \Call{Expanderization}{}}
    \end{subfigure}
    \quad 
    \begin{subfigure}[b]{0.4\textwidth}
        \centering
        \begin{tikzpicture}[scale=0.8]

        \node[circle, draw=black!80, fill=green!9, minimum size=1pt, thick] at (1.2,1.3)(v1){\small \phantom{a}};

        \node[circle, draw=black!80, fill=green!9, minimum size=1pt, thick] at (-0.3,-1.4)(v2){\small \phantom{a}};

        \node[circle, draw=black!80, fill=green!9, minimum size=1pt, thick] at (0,2.3)(v5){\small \phantom{a}};

        \node[circle, draw=black!80, fill=green!9, minimum size=1pt, thick] at (0,-2.5)(n2){\small \phantom{a}};

        \node[circle, draw=black!80, fill=green!9, minimum size=1pt, thick] at (1.5,-0.8)(v3){\small \phantom{a}};

        \node[circle, draw=black!80, fill=green!9, minimum size=1pt, thick] at (2.5,-1.5)(n3){\small \phantom{a}};

        \node[circle, draw=black!80, fill=green!9, minimum size=1pt, thick] at (-1.,0.5)(v4){\small \phantom{a}};

        \node[circle, draw=black!80, fill=green!9, minimum size=1pt, thick] at (-2.5,1.5)(n4){\small \phantom{a}};

        \draw[very thick, draw=black!80] (v2) --(n2);

        \draw[very thick, draw=black!80] (v5) --(v4);

        \draw[very thick, draw=black!80] (n2) --(n3);

        \draw[very thick, draw=black!80] (v4) --(n4);

        \draw[very thick, draw=black!80] (v3) --(n3);

        \draw[very thick, draw=black!80] (v2) --(v5);

        \draw[very thick, draw=black!80] (v1) --(v3);

        \draw[very thick, draw=red!80] (v1) --(v2);
        \draw[very thick, draw=red!80] (v4) --(v1);

        \draw[very thick, draw=red!80] (v2) --(n3);

        \draw[very thick, draw=red!80] (n4) --(v5);

        \draw[very thick, draw=red!80] (v3) --(v5);

        \draw[very thick, draw=red!80] (v3) --(n2);

        \draw[very thick, draw=red!80] (v3) --(v4);

        \draw[very thick, draw=red!80] (v2) --(n4);

    \end{tikzpicture}
    \caption{$G$ after \Call{Expanderization}{}}
    \end{subfigure}
    \caption{A graph $G$ before and after \Call{Expanderization}{}. The edges of the expander are marked in red.}
    \label{fig:expanderize}
\end{figure}
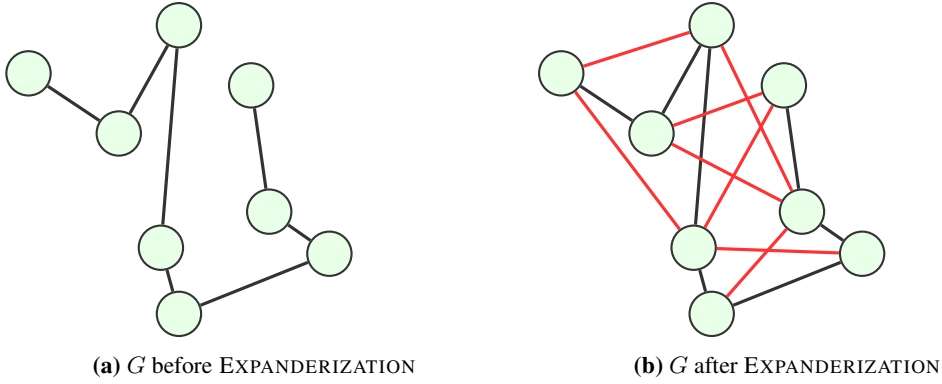

Let us now record some parameters of $G'$.
The new graph is regular with degree $d + d_0$ and the new number of edges is $n(d + d_0)/2$.
Also, the new constraint graph is indeed a constant degree expander because the new $\lambda'$ is at most $d + \lambda_0 < d + d_0$ by \Cref{lem:expander-superimpose}.
Note that domain $\Sigma$ does not change and so $\Sigma'=\Sigma$.

We now show how the sensitivity bound translates through the transformation.
\begin{lemma}\label{lem:expanderization}
    Let $\mathcal I$ be a swap-closed family of binary $d$-regular instances of a CSP over $\Sigma$, and let $\mathcal I' = \Call{Expanderize}{\mathcal I}$.
    Then for any $\varepsilon>0$, we have
    \[
        \swapsens_{1,1-\varepsilon'}(\swapclo(\mathcal I')) \geq \swapsens_{1,1-\varepsilon}(\mathcal I)
    \] 
    for $\varepsilon' := \varepsilon/C$, where $C>0$ is a universal constant.
\end{lemma}
\begin{proof}
    Let $I = (V, E, \Sigma, \mathcal R) \in \mathcal I$ and $I' = (V,E', \Sigma, \mathcal R') = \Call{Expanderize}{\mathcal I}$.
    We consider a trivial algorithm $T_\sigma$ that, given an assignment $\bm \sigma'$ for $I'$, outputs $\bm \sigma = \bm \sigma'$ as an assignment for $I$.

    We show that the pair $(\Call{Expanderize}{},T_\sigma)$ is a $(1,1-\varepsilon,c'=1,s' = 1-\varepsilon',C_I=1,C_\sigma=1)$-sensitivity-preserving reduction.
    Then the claim follows by \Cref{lem:sensitivity-reduction}.
    As the analysis for  $c'=1$, $C_I=1$, and $C_\sigma=1$ is trivial, we analyze $s'$ below.
    Let $\bm \sigma'$ be an assignment for $I'$ with $\E \val_{I'}(\bm \sigma') \geq 1-\varepsilon' = s'$.
    Then we have 
    \[
        \varepsilon' m' \geq \E \cost_{I'}(\bm \sigma) \cdot m' = \E \cost_{I}(\bm \sigma) \cdot m
        = \E \cost_{I}(\bm \sigma) \cdot \frac{d m'}{d+d_0}.
    \]
    Hence, we get $\E \cost_{I}(\bm \sigma) \leq \varepsilon m$ by setting $C \geq d/(d+d_0)$.
\end{proof}

\subsection{Gap Amplification}\label{subsec:gap-amplification}

In this section, we consider amplifying the gap between the completeness and soundness.
Here, we use the transformation based on graph powering, introduced by Dinur~\cite{dinur2007pcp}.

\subsubsection{Graph Powering}
Let $t$ be an integer and $I:=(V,E, \Sigma, \mathcal{R}=\{R_e\}_{e \in E})$ be a CSP instance, where the graph $G=(V,E)$ is an $(n,d,\lambda)$-expander.
We construct a new instance $I':=(V,E',\Sigma,\mathcal R'=\{R'_e\}_{e \in E})$, where the constraint $R'_e$ for $e \in E$ is a subset of $\Sigma^{1+d+\cdots +d^t} \times \Sigma^{1+d+\cdots +d^t}$. 
As $G$ is $d$-regular, for any vertex $v \in V$, the number of vertices at distance at most $t$ from $v$ is at most $1+d+\cdots +d^t$. 
In any assignment $\sigma' : V \rightarrow \Sigma^{1+d+\cdots +d^t}$ for $I'$, we will think of each vertex $v \in V$ as having an ``opinion'' about what the value of each vertex $w \in V$ at distance $\leq t$ should be; which is $\sigma'(v)_w$. 
Roughly, the constraints will state that if the edge $(a,b)$ is of distance $\leq t$ from $u$ and $v$ then the assignment $(\sigma'(v)_a, \sigma'(w)_b)$ satisfies the constraint $R_{ab}$. 
As each vertex now is required to answer correctly on all constraints within radius $t$, this should blow up the gap by $O(t)$. 

It will be convenient to define the edges $E'$ and the constraints $\{R'_e\}_{e \in E'}$ by the following process:
\begin{enumerate}
    \item Describe a distribution over edges $e \in E'$ and corresponding constraints $R'_e$.
    \item View this distribution as a weighted graph, where edges have rational weights $w_e$ depending on $d$, $|\Sigma|$, and $t$.
    \item Introduce copies of the constraint $R_e$, where the number of copies is proportional to $w_e$.
\end{enumerate}
In particular, the distribution will be useful for analyzing the change in the gap and the sensitivity. In order to describe this distribution, we will make use of following two types of random walks.
\begin{itemize}
    \item[--] \emph{Before Stopping Random Walk (BSRW).}~ Pick a random vertex $\bm v \in V$ to be the start vertex. Repeat the following: Choose a random neighbour of the current vertex, and move to that vertex, and halt with probability $1/t$.
    \item[--] \emph{After Stopping Random Walk (ASRW).}~ Let $v$ be a given start vertex. Repeat the following: With probability $1/t$ halt. If we did not halt, pick a random neighbour of the current vertex, and move to that vertex. 
\end{itemize}

We now describe the gap amplification step, which is by a reduction known as \Call{Powering}{}.
Given an instance $I=(V,E,\Sigma,\mathcal{R}=\{R_e\}_{e \in E})$, the set $E'$ and the constraints $\{R'_e\}_{e \in E'}$ of the instance $I'$ that we are constructing are defined by the following distribution:
\begin{itemize}
    \item $E'$:  Pick a random vertex $\bm v$. Perform an ASRW from $\bm v$, ending at some vertex $\bm w$. 
    If the length of this walk is greater than $B:=10t\log |\Sigma|$, then do nothing. 
    Otherwise, add an edge $(\bm v,\bm w)$ to $E'$.
    \item $R'_e$: 
    Let $e = (v,w)$ be the edge defined in the previous step. We describe the constraint $R'_e$: for each edge $(a,b)$ traversed in the ASRW used to define $e$, if $d_G(v,a) \leq t$ and $d_G(w,b) \leq t$, then add the constraint $(\sigma'(v)_a, \sigma'(w)_b) \in R_{ab}$ to $R'_{uv}$.
\end{itemize}
We note that the probability that an edge $e = (v,w)$ is chosen in the above distribution is a multiple of $1/(n (dt)^B)$ and we can view the distribution as a weighted CSP instance, where each weight $w_e$ is a multiple of $1/(n(dt)^B)$.
Then, we output the CSP instance $I'$ obtained by copying each constraint $n(dt)^B w_e$-many times.
Note that the underlying graph of $I'$ is $D:=(dt)^B$-regular.

It remains to show that the gap increases by $O(t)$.
Although the sensitivity lower bound decreases by $D$ here, we will later offset this by applying degree reduction.
\begin{lemma}\label{lem:powering}
    Let $\mathcal I$ be a swap-closed family of binary instances of a CSP over $\Sigma$ such that the underlying graph is an $(n,d,\lambda)$-expander, and let $\mathcal I' = \Call{Powering}{\mathcal I}$.
    Then for any $\varepsilon = O(1/t)$, we have 
    \[
        \swapsens_{1,1-\varepsilon'}(\swapclo(\mathcal I')) \geq \frac{1}{8D} \cdot \swapsens_{1,1-\varepsilon}(\mathcal I),
    \] 
    for 
    \[
       D := (dt)^B, \quad
       B := 10t\log |\Sigma|,
       \quad \text{and} \quad
       \varepsilon' := \frac{t}{C |\Sigma|^4} \cdot \varepsilon,
    \]
    where $C>0$ is a universal constant.
\end{lemma}

Let  $I = (V, E, \Sigma, \mathcal R) \in \mathcal I$ and  $I' = (V,E',\Sigma,\mathcal R') = \Call{Powering}{\mathcal I}$.
We design an algorithm $T_\sigma$ that extracts a solution $\bm \sigma$ for the original instance $I$ from an assignment $\bm \sigma': V' \to \Sigma$ for $I'$.
This will be done by another probability distribution, using a random walk and truncating probabilities.
Fix a vertex $v \in V$ and consider the following: perform a BSRW starting from $v$, conditioned on this walk ending within $t$ steps. 
Let $\bm w$ be the final vertex. 
This gives a probability distribution $\mu_v: \Sigma \rightarrow \mathbb{R}$ over opinions $\sigma'(\bm w)_v$ of what $v$'s value should be. 
That is, $\bm w$ contributes the value $\sigma'(\bm w)_v$ towards what $v$'s value should be, weighted by the probability of ending at $\bm w$ in $t$ steps given that we started at $v$. 
Then, we consider the \emph{truncated} distribution $\mu^*_v$, which is defined as
\[
    \mu_v^*(a) = \frac{\qty[\mu_v(a) - \frac{1}{10|\Sigma|}]_+}{\sum_{b \in \Sigma} \qty[\mu_v(b) - \frac{1}{10|\Sigma|}]_+}.
\]
In particular $\mu_v^*(a) >0$ only if $\mu_v(a) > 1/(10|\Sigma|)$.
The value of $\bm \sigma(v)$ is drawn from the distribution $\mu_v^*$.
Our goal is to show that the pair $(\Call{Powering}{},T_\sigma)$ is $(1,1-\varepsilon,c'=1,s'=1-\varepsilon',C_I=D,C_\sigma=8)$-sensitivity-preserving reduction, where $\varepsilon'$ is as in \Cref{lem:powering}.
Then, \Cref{lem:powering} follows by \Cref{lem:sensitivity-reduction}.
The analysis for $c'$ and $C_I$ is clear, and hence we focus on analyzing $s'$ and $C_\sigma$ and discuss them in \Cref{subsubsec:gap-amplification-sensitivity-increase,subsubsec:gap-amplification-gap-analysis}, respectively.

\subsubsection{Sensitivity Increase in the Recovery Procedure}\label{subsubsec:gap-amplification-sensitivity-increase}
In this we section, we show the following:
\begin{lemma}\label{lem:powering-sensitivity}
    The choice $C_\sigma = 8$ satisfies \Cref{item:reduction-assignment-distance} of \Cref{def:sensitivity-reduction}.
\end{lemma}
\begin{proof}
    Let $I \in \mathcal I$ be an instance and $I'$ be the new instance after \Call{Powering}{}.
    Let $\sigma',\tilde \sigma'$ be two assignments for $I'$ with $\Ham(\sigma',\tilde \sigma')=1$.
    For $v \in V$, let $\mu_v,\tilde \mu_v$ be the distributions constructed as before from $\sigma'$ and $\tilde \sigma'$, and similarly for $\mu_v^*,\tilde \mu_v^*$.
    Let $\bm \sigma,\tilde {\bm \sigma}$ be assignments for $I$ constructed from $\sigma',\tilde \sigma'$, respectively.
    We have    
    \[
        \EMD(\bm \sigma, \tilde{\bm \sigma}) 
        \leq \sum_{v \in V} \TV(\bm \sigma(v), \tilde{\bm \sigma}(v)) 
        \leq \sum_{v \in V} \|\mu^*_v - \tilde \mu^*_v \|_1. 
    \]

    Let $A_v = \sum_{a \in \Sigma}\qty[\mu_v(a) - \frac{1}{10|\Sigma|}]_+$ and $\tilde A_v = \sum_{a \in \Sigma}\qty[\tilde \mu_v(a) - \frac{1}{10|\Sigma|}]_+$.
    Let $\delta_{v,a} := \tilde \mu_v(a) - \mu_v(a)$.
    Note that $A_v \geq 9/10$, $\tilde A_v \geq 9/10$ and 
    \begin{align}
        & \big|A_v - \tilde A_v \big|
        \leq \sum_{a \in \Sigma} \qty |\qty[\mu_v(a) - \frac{1}{10|\Sigma|}]_+ 
        - \qty[\tilde \mu_v(a) - \frac{1}{10|\Sigma|}]_+ | \nonumber \\
        & \leq \sum_{a \in \Sigma} \qty |\qty[\mu_v(a) - \frac{1}{10|\Sigma|}]_+ 
        - \qty[\mu_v(a) + \delta_{v,a} - \frac{1}{10|\Sigma|}]_+ | \nonumber \\   
        & \leq \sum_{a \in \Sigma} |\delta_{v,a}|
        = \|\mu_v - \tilde \mu_v\|_1. \label{eq:powering-sensitivity-1}
    \end{align}
    Then, we have
    \begin{align*}
        & \|\mu_v^* - \tilde \mu_v^*\|_1
        = \sum_{a \in \Sigma} \qty|\frac{\qty[\mu_v(a) - \frac{1}{10|\Sigma|}]_+ }{A_v} - \frac{\qty[\tilde \mu_v(a) - \frac{1}{10|\Sigma|}]_+}{\tilde A_v}| \\
        & = \sum_{a \in \Sigma} \qty|\frac{\qty[\mu_v(a) - \frac{1}{10|\Sigma|}]_+ }{A_v} - \frac{\qty[\mu_v(a) + \delta_{v,a} - \frac{1}{10|\Sigma|}]_+}{\tilde A_v}| \\
        & = \sum_{a \in \Sigma} \qty|\frac{\tilde A_v \cdot \qty[\mu_v(a) - \frac{1}{10|\Sigma|}]_+ - A_v \cdot \qty[\mu_v(a) + \delta_{v,a} - \frac{1}{10|\Sigma|}]_+}{A_v \tilde A_v}| \\
        & \leq \sum_{a \in \Sigma} \qty|\frac{(A_v \pm \|\mu_v - \tilde \mu_v\|_1) \cdot \qty[\mu_v(a) - \frac{1}{10|\Sigma|}]_+ - A_v \cdot \qty[\mu_v(a) + \delta_{v,a} - \frac{1}{10|\Sigma|}]_+}{A_v \tilde A_v}|       \tag{by \eqref{eq:powering-sensitivity-1}} \\
        & \leq 
        \sum_{a \in \Sigma} \qty|\frac{\qty[\mu_v(a) - \frac{1}{10|\Sigma|}]_+ - \qty[\mu_v(a) + \delta_{v,a} - \frac{1}{10|\Sigma|}]_+}{\tilde A_v}|  
        + 
        \sum_{a \in \Sigma} \qty|\frac{\|\mu_v - \tilde \mu_v\|_1 \cdot \qty[\mu_v(a) - \frac{1}{10|\Sigma|}]_+ }{A_v \tilde A_v}| \\
        & \leq 
        \frac{\sum_{a \in \Sigma} |\delta_{v,a}|}{\tilde A_v} 
        + 
        \frac{\|\mu_v - \tilde \mu_v\|_1 }{\tilde A_v}  \\
        & \leq 4\|\mu_v - \tilde \mu_v\|_1,
    \end{align*}
    where the last inequality is because $\tilde A_v \geq 9/10$.
    Let $w \in V$ be such that $\sigma'(w) \neq \tilde \sigma'(w)$ and let $p_{v \rightarrow w}$ be the probability that we reach $w$ from a vertex $v \in V$ in the BSRW. 
    Note that $w$ contributes by $2p_{v \rightarrow w}$ to the value of $\|\mu_v - \tilde \mu_v\|_1$.
    Then, we have
    \[ 
        \sum_{v \in V} \|\mu_v^* - \tilde \mu_v^* \|_1 
        \leq 4 \sum_{v \in V} \|\mu_v - \tilde \mu_v \|_1       
        \leq 8 \sum_{v \in V} p_{v \rightarrow w} = 8.
        \qedhere
    \]
\end{proof}

\subsubsection{Gap Analysis}\label{subsubsec:gap-amplification-gap-analysis}
We now prove that the gap increases significantly when $t$ is large (but still constant).
\begin{lemma}\label{lem:powering-gap-increase}
    Suppose $\varepsilon = O(1/t)$.
    Then, the choice $s' = 1-\varepsilon'$ satisfies \Cref{item:reduction-soundness} of \Cref{def:sensitivity-reduction} for 
    \[
        \varepsilon' = \frac{t}{C |\Sigma|^4} \cdot \varepsilon,
    \]
    where $C>0$ is a universal constant.
\end{lemma}
Suppose $\sigma':V \to \Sigma^{1+d+\cdots+d^t}$ satisfies $\val_{I'}(\sigma') \geq 1-\varepsilon' = s'$.
Our goal is to show that $\E \val_I(\bm \sigma) \geq 1-\varepsilon$.
We can handle the situation that we have a random assignment $\bm \sigma'$ for $I'$ by the conditioning argument as in the proof of \Cref{lem:degree-reduction}.

Let $\bm \sigma$ be the extracted assignment and let $F_\sigma \subseteq E$ be the set of edges in $G=(V,E)$ whose constraints are violated when $\bm \sigma = \sigma$. 
We relate the expected number of constraints violated by $\bm \sigma$ to the number constraints violated by $\sigma'$ using the notion of a ``faulty'' step in the ASRW.

\begin{definition}[Faulty step]
    For an assignment $\sigma : V \to \Sigma$, a $\sigma$-faulty step in the ASRW defining an edge $e' = (x,y) \in E'$ is an edge $(u,v) \in E$ along this path satisfying 
    \begin{enumerate}
        \item[(i)] $(u,v) \in F_\sigma$
        \item[(ii)] $d_G(x,u) \leq t$ and $\sigma'(x)_u = \sigma(u)$
        \item[(iii)]  $d_G(y,v) \leq t$ and $\sigma'(y)_v = \sigma(v)$
    \end{enumerate}        
    We further define a step to be $\sigma$-faulty$^*$ if
    \begin{enumerate}
    \item the step is faulty,
    \item the number of steps in the overall walk is at most $B$. 
    \end{enumerate}     
\end{definition}
Let $\bm N_\sigma$ and $\bm N_\sigma^*$ be the numbers of $\sigma$-faulty and $\sigma$-faulty$^*$ steps, respectively, in the ASRW with respect.
Let $\bm S$ be the total number of steps.
By definition, we have $\bm N_\sigma^* = \bm N_{\sigma} \cdot \bm 1[\bm S\leq B]$.
We will use this to bound $\varepsilon'$ as follows
\begin{align}
   \varepsilon' \geq \frac{|\{(u,v) \in E' : (\sigma'(u),\sigma'(v)) \not \in R_{uv}'  \}|
   }{|E'|} 
   = \Pr_{e \sim E'}[ \sigma' \mbox{ violates } R'_{e}]  
   \geq \Pr[\bm N_{\bm \sigma}^*>0]  
   \geq \frac{\E[\bm N_{\bm \sigma}^*]^2}{\E[(\bm N_{\bm \sigma}^*)^2]}, \label{eq:secondmom}
\end{align}
where the final inequality follows by the second moment method.
Let $\bar F := \E F_{\bm \sigma}$.
We will later show the following two lemmas.
\begin{lemma}\label{lem:powering-gap-increase-1}
    The following holds:
    \[
        \E[ \bm N_{\bm \sigma}^* ] \geq \frac{t \bar F}{1600 |\Sigma|^2m}.        
    \]
\end{lemma}
\begin{lemma}\label{lem:powering-gap-increase-2}
     The following holds: 
    \[
        \E[(\bm N_{\bm \sigma}^*)^2] 
        \leq 
        \left(1 + \frac{1}{1-\lambda/d} \right) \cdot \frac{t \bar F}{m} + \frac{2t^2}{m^2} \qty((2d+1)\bar F + \bar F^2 ).
    \]
\end{lemma}
\begin{proof}[Proof of \Cref{lem:powering-gap-increase}]    
    Suppose the first term in \Cref{lem:powering-gap-increase-2} is smaller than the second term. 
    Then, we have
    \begin{align*}
        & \left(1 + \frac{1}{1-\lambda/d} \right) \cdot \frac{t \bar F}{m} < \frac{2t^2}{m^2} \qty((2d+1)\bar F + \bar F^2 ) \\
        \Leftrightarrow & \left(1 + \frac{1}{1-\lambda/d} \right) < \frac{2t}{m} \qty((2d+1) + \bar F ) \\
        \Leftrightarrow & \E \cost_I(\bm \sigma) > \frac{1}{2t}\left(1 + \frac{1}{1-\lambda/d} \right) - \frac{2d+1}{m}.
    \end{align*}
    This is a contradiction from the condition that $\E \cost_I(\bm \sigma) = O(1/t)$ (by choosing the hidden constant to be small enough).

    Now, the first term in \Cref{lem:powering-gap-increase-2} is larger than or equal to the second term.
    Then combining \Cref{lem:powering-gap-increase-1,lem:powering-gap-increase-2}, we obtain from \eqref{eq:secondmom}
    \[
        \varepsilon'
        \geq \E_{\sigma \sim \bm \sigma}
        \qty[
            \frac{1}{2 \cdot 1600^2 |\Sigma|^4 \left(1 + \frac{1}{1-\lambda/d} \right)}
            \cdot \frac{ t \E|F_{\bm \sigma}| }{m} 
        ]
        \geq \frac{t \E \cost_I(\bm \sigma) }{C |\Sigma|^4},
    \]
    where $C>0$ is a large enough constant.
    By setting $\varepsilon' =  t/(C |\Sigma|^4) \cdot \varepsilon$, we obtain $\E \cost_I(\bm \sigma) \leq \varepsilon$.
\end{proof}

We will use the following fact in order to analyze $N_{\bm \sigma}$.

\begin{fact}\label{fact:ASRWtoBSRW}
    Consider an ASRW in a graph $G$, conditioned on there being exactly $k$ $u \rightarrow v$ steps. 
    Let $\bm x,\bm y$ be the initial and final vertices of the walk. Then $\bm x$ and $\bm y$ are independent random variables, where $\bm x$ (resp., $\bm y$) is distributed as a BSRW from $u$  (resp., $v$).
\end{fact}

\begin{proof}[Proof of \Cref{lem:powering-gap-increase-1}]
    We first bound $\bm N_{\bm \sigma}$.
    Consider conditioning $\bm \sigma = \sigma$ and let $(u,v) \in F_\sigma$ be a violated edge.
    Then, we have
    \begin{align} 
        & \E[\bm N_\sigma]
        = 
        \E[\# \mbox{ $\sigma$-faulty $u\rightarrow v$ steps}] \nonumber \\ 
        & = \sum_{k \geq 1} \E[\# \text{ $\sigma$-faulty $u\rightarrow v$ steps} \mid \text{exactly $k$ $u \rightarrow v$ steps}] \cdot \Pr[\text{exactly $k$ $u \rightarrow v$ steps}] \nonumber \\
        & = \sum_{k \geq 1} k \Pr[u \rightarrow v \mbox{ step is $\sigma$-faulty}\mid \text{exactly $k$ $u \rightarrow v$ steps}] \cdot \Pr[\text{exactly $k$ $u \rightarrow v$ steps}],      
        \label{eq:c2eq1}  
    \end{align}
    where the last equality holds because either all $u \rightarrow v$ steps are $\sigma$-faulty or not.

    \begin{claim}
        We have 
        \[
            \Pr[u \rightarrow v \text{ step is $\sigma$-faulty} \mid     \text{exactly $k$ $u \rightarrow v$ steps}]        
            \geq \frac{1}{400|\Sigma|^2}.
        \]
    \end{claim}
    \begin{proof}
        Suppose that the ASRW makes exactly $k$ $u \rightarrow v$ steps. 
        As $(u,v) \in F_\sigma$, this step is faulty if (ii) and (iii) hold.
        As $\bm x$ and $\bm y$ are independent by \Cref{fact:ASRWtoBSRW},  we have $\Pr[(ii) \mbox{ and } (iii)] = \Pr[(ii)] \cdot \Pr[(iii)]$
        and as $\Pr[(ii)]=\Pr[(iii)]$, it is enough to show that $\Pr[(ii)] \geq \frac{1}{20|\Sigma|}$.
        Let $\bm x$ be a random vertex generated by taking a BSRW from $u$, and let $\bm \ell$ be the number of steps of the walk. 
        Then by \Cref{fact:ASRWtoBSRW},
        \begin{align*} 
            & \Pr[d_G(u,\bm x) \leq t \mbox{ and } \sigma'(\bm x)_u = \sigma(u)] \\ 
            & = \Pr[d_G(u,\bm x) \leq t \mbox{ and } \sigma'(\bm x)_u =  \sigma(u) \mid \bm \ell \leq t] \cdot \Pr[ \bm \ell \leq t ] \\
            &= \Pr[\sigma'(\bm x)_u = \sigma(u) \mid \bm \ell \leq t] \cdot \Pr[ \bm \ell \leq t] \\
            & \geq \frac{1}{2} \cdot \Pr[\sigma'(\bm x)_u = \sigma(u) \mid \bm \ell \leq t] \\
            & = \frac{\mu_u(\sigma(u))}{2} \\
            & \geq \frac{1}{20|\Sigma|},
        \end{align*}
        where the first inequality holds because the BSRW halts within $t$ steps with probability $1-(1-1/t)^t \geq 1/2$,
        the last equality holds because the distribution on $\bm x$ is precisely $\mu_u$ --- take a BSRW from $u$, conditioned on stopping within $t$ steps, and the last inequality holds because $\sigma(u)$ is in the support of $\mu_u^*$, which implies that $\mu_u(\sigma(u)) \geq \frac{1}{10|\Sigma|}$.        
    \end{proof}
    By the claim above, we have
    \[
        \eqref{eq:c2eq1} \geq \sum_{k \geq 1} \frac{k}{400|\Sigma|^2} \cdot \Pr[\text{exactly $k$ $u \rightarrow v$  steps}]  
        = \frac{t}{800 |\Sigma|^2 m },
     \]
    where the last equality follows because each step is equally likely to be one of the $2m$ possibilities and the expected total number of steps is $t$.
    Hence, we have 
    \[
        \E[\bm N_{\sigma}] \geq \frac{t |F_\sigma|}{800 |\Sigma|^2 m }
    \]
    and unconditioning $\bm \sigma$, we have
    \[
        \E[\bm N_{\bm \sigma}] \geq \frac{t \bar F}{800 |\Sigma|^2 m }.
    \]

    Next, we bound $\bm N^*_{\bm \sigma}$.
    Note that
    \begin{align}
         \E[\bm N_{\bm \sigma}^*] &= \E[\bm N_{\bm \sigma} \cdot \bm 1[\bm S\leq B]]
        = \E[\bm N_{\bm \sigma} \cdot  (1- \bm 1[\bm S>B)]] \nonumber \\
        & = \E[\bm N_{\bm \sigma}] - \E[\bm N_{\bm \sigma} \cdot \bm 1[\bm S>B]] \nonumber\\
        &\geq \frac{t \bar F}{800 |\Sigma|^2 m} - \E\qty[\bm N_{\bm \sigma} \cdot \bm 1[\bm S>B]],  
        \label{eq:N_sigma^*}
    \end{align}
    We can bound the second term as
    \begin{align*} 
        & \E\qty[\bm N_{\bm \sigma} \cdot \bm 1[\bm S>B]]  =\Pr[\bm S>B] \cdot \E[\bm N_{\bm \sigma} \mid \bm S>B] 
        = \qty(1-\frac{1}{t})^B \cdot \E[\bm N_{\bm \sigma} \mid \bm S>B] \\ 
        & \geq \exp \qty(-\frac{e B}{(e-1)t}) \cdot \E[\bm S \mid \bm S>B] \cdot \frac{\bar F}{m}
        \tag{by $e^{-x} \leq 1 - \frac{e-1}{e}x$ for $x \in [0,1]$} \\
        & = \exp \qty(-\frac{eB}{(e-1)t}) \cdot(B+t) \cdot \frac{\bar F}{m} \\ 
        & \geq \frac{1}{|\Sigma|^{10e/(e-1)}} \cdot(20 t \log |\Sigma| ) \cdot \frac{\bar F}{m}  \\
        & \geq \frac{t \bar F}{1600|\Sigma|^2 m }, 
    \end{align*}
    where we assume $|\Sigma|$ is large enough in the last inequality.
    Finally, by~\eqref{eq:N_sigma^*}, we conclude that
    \[
        \E[\bm N_{\bm \sigma}^*] \geq \frac{t \bar F}{1600 |\Sigma|^2 m}. \qedhere
    \]    
\end{proof}

To prove \Cref{lem:powering-gap-increase-2}, we use the following result, which states that the second moment of $F_{\bm \sigma}$  is not too large compared to its first moment.
\begin{lemma}\label{lem:F_sigma-moment}
    We have
    \[
        \E |F_{\bm \sigma}|^2 \leq (2d+1) \bar F + \bar F^2.
    \]
\end{lemma}
\begin{proof}
    For an edge $e \in E$, let $\bm X_e$ be the indicator of the event that $\bm \sigma$ violates the constraint $(e,R_e)$, and let $p_e = \E X_e$ be the probability of the event.
    For two edges $e,f \in E$, we write $e \sim f$ to denote that they are incident.
    Note that $\bm X_e$ and $\bm X_f$ are independent if $e \not \sim f$.
    Then, we have
    \begin{align*}
        & \E |F_{\bm \sigma}|^2 
        = \E \qty( \sum_{e \in E}\bm X_e)^2 \\
        & = \E \sum_{e \in E} \bm X_e^2 + \E \sum_{e,f \in E: e \not \sim f} \bm X_e \bm X_f + \E \sum_{e,f \in E: e  \sim f} \bm X_e \bm X_f \\
        & \leq \E |F_{\bm \sigma}|
        + \qty(\E \sum_{e \in E} \bm X_e)^2
        + \sum_{e,f \in E: e  \sim f} \max\{p_e,p_f\} \\
        & \leq \E |F_{\bm \sigma}| + \qty(\E |F_{\bm \sigma}|)^2 + 2d \sum_{e \in E}p_e \\
        & \leq (2d+1)\E |F_{\bm \sigma}| + \qty(\E |F_{\bm \sigma}|)^2 \\
        & \leq (2d+1) \bar F + \bar F^2.
        \qedhere
    \end{align*}
\end{proof}

\begin{proof}[Proof of \Cref{lem:powering-gap-increase-2}] 
    Although the proof is similar to that of Lemma~5.2 in~\cite{dinur2007pcp}, we need to account for the fact that $\bm \sigma$ is a random variable.

    Consider conditioning on $\bm \sigma = \sigma$.
    Let $\bm S_i$ be the indicator random variable which is $1$ iff the $i$-th step of the ASRW is in $F_\sigma$, and let $\bm M:= \sum_{i=1}^\infty \bm S_i$ be the number of steps of the ASRW that were within $F_\sigma$. 
    Then, we have
    \begin{align*}
        \E[(\bm N^*_\sigma)^2] &\leq \E[\bm M^2] 
        =\sum_{i,j =1}^\infty \E[\bm S_i \bm S_j] \\
        &\leq 2 \sum_{i=1}^\infty \Pr[\bm S_i=1] \cdot \sum_{j \geq i} \Pr[\bm S_j=1 \mid \bm S_i=1] \\
        &=2 \sum_{i=1}^\infty \Pr[\bm S_i=1] \left( \Pr[\bm S_i=1 \mid \bm S_i=1] + \sum_{j >i} \Pr[\bm S_j=1 \mid \bm S_i=1] \right) \\
        &=2 \sum_{i=1}^\infty \Pr[\bm S_i=1] \left(1 + \sum_{j >i} \Pr[\bm S_j=1 \mid \bm S_i=1] \right).
    \end{align*}
    Now, 
    \begin{align*}  
        \Pr[\bm S_j=1 \mid \bm S_i=1]  &= \Pr[\mbox{The ASRW takes $j-i$ more steps}] \cdot \Pr[\mbox{$(j-i)$-th step is in $F_\sigma$}] \\
        &\leq \qty(1-\frac{1}{t})^{j-i}  \left( \frac{|F_\sigma|}{m} + \qty(\frac{\lambda}{d})^{\ell-1}\right),
    \end{align*}
    by Proposition~5.4 of~\cite{dinur2007pcp}), as $G$ is a $(n,d, \lambda)$-expander. 

    Substituting this into the previous bound, we have
    \begin{align*}
        \E[(\bm N_\sigma^*) ^2] &\leq 2 \sum_{i=1}^\infty \Pr[\bm S_i=1] \left(1+  \sum_{\ell=1}^\infty \qty(1-\frac{1}{t})^\ell \left( \frac{ |F_\sigma|}{m}+\qty(\frac{\lambda}{d})^{\ell-1} \right) \right) \\
        &\leq 2 \sum_{i=1}^\infty \Pr[\bm S_i=1] \left(1 +(t-1) \frac{ |F_\sigma|}{m} + \sum_{\ell=1}^\infty \qty(\frac{\lambda}{d})^{\ell-1} \right) \tag{Since $1-1/t<1$} \\
        &\leq 2 \sum_{i=1}^\infty \Pr[\bm S_i=1] \left(1 + \frac{t|F_\sigma|}{m} + \frac{1}{1-\lambda/d} \right) \tag{Since $\lambda < d$}\\
        &= 2\left(1 + \frac{t|F_\sigma|}{m} + \frac{1}{1-\lambda/d} \right) \E[\bm M] \\
        &\leq 2\left(1 + \frac{t|F_\sigma|}{m} + \frac{1}{1-\lambda/d} \right) \cdot \frac{t  |F_\sigma|}{m}.
    \end{align*}
    By unconditioning $\bm \sigma$, we obtain 
    \begin{align*}
        & \E[(\bm N_{\bm \sigma}^*) ^2] 
        \leq 2 \left(1 + \frac{1}{1-\lambda/d} \right) \cdot \frac{t \bar F}{m} + \frac{2t^2}{m^2} \E |F_{\bm \sigma}|^2 \\ 
        & \leq \left(1 + \frac{1}{1-\lambda/d} \right) \cdot \frac{t \bar F}{m} + \frac{2t^2}{m^2} \qty((2d+1)\bar F + \bar F^2 ),
        \tag{by \Cref{lem:F_sigma-moment}}
    \end{align*}
    which completes the proof.
\end{proof}

\subsubsection{Sensitivity Recovery}

The gap amplification step decreases the sensitivity bound by $D=(dt)^B$, but we can offset it by combining it with degree reduction:
\begin{lemma}\label{lem:powering+degree-reduction}
    Let $\mathcal I$ be a swap-closed family of instances of a binary CSP such that the underlying graph is an $(n,d,\lambda)$-expander, and let $\mathcal I' = \Call{DegreeReduction}{\Call{Powering}{\mathcal I}}$.
    Then for any $\varepsilon = O(1/t)$, we have 
    \[
        \swapsens_{1,1-\varepsilon'}(\swapclo(\mathcal I')) \geq \Omega\qty( \swapsens_{1,1-\varepsilon}(\mathcal I)),
    \] 
    for 
    \[
       \varepsilon' := \frac{t}{C |\Sigma|^4} \cdot \varepsilon,
    \]
    where $C>0$ is a universal constant.
\end{lemma}
\begin{proof}

    The claim follows by combining \Cref{lem:degree-reduction,lem:powering} and noting that every instance in $\Call{Powering}{\mathcal I}$ is $D$-regular, where $D$ is as in the statement of \Cref{lem:powering}.
\end{proof}
\subsection{Alphabet Reduction}\label{subsec:alphabet-reduction}

The previous step leaves us with a set of constraints over an alphabet of size $|\Sigma|^{1+d+ \ldots +d^t}$. 
In this section, we introduce a procedure \Call{AlphabetReduction}{} which decreases the alphabet size while preserving the gap, and not decreasing the sensitivity by too much. This procedure is identical to the construction of Dinur~\cite{dinur2007pcp}. Our contribution is a new procedure $T_\sigma$ which recovers an assignment to the instance prior to \Call{AlphabetReduction}{} from the instance after. This new procedure will allow us to show that $(\Call{AlphabetReduction}{}, T_\sigma)$ is a sensitivity-preserving reduction. 

We begin by recalling the \Call{AlphabetReduction}{} of Dinur. The heart of which is the sub-routine \Call{AssignmentTester}{}, which will be ran on each edge of the original instance. 

\subsubsection{Assignment Tester}
We say that two assignments $x \in \{0,1\}^k$ and $y \in \{0,1\}^k$ are \emph{$\delta$-far} if $\Ham(x,y) \geq \delta \cdot k$.

\begin{lemma}[Assignment Tester~\cite{dinur2007pcp}]\label{lem:assignTest}
     There is a reduction which takes as input a Boolean circuit $C:\{0,1\}^X \to \{0,1\}$ over Boolean variables $X$ and outputs a binary CSP instance $I:=(X \cup T,E, \Sigma_0, {\cal R})$ over Boolean variables $X$ and an additional set of variables $T$ over an alphabet of size $|\Sigma_0| \leq 2^6$.
     Furthermore, for every assignment $\alpha \in \{0,1\}^X$:
    \begin{itemize}
        \item If $C(\alpha)=1$, then there exists $\beta \in \Sigma_0^T$ such that $(\alpha,\beta)$ satisfies every constraint in $I$, otherwise
        \item If $\alpha$ is $\delta$-far from every $\alpha^* \in \{0,1\}^X$ for which $C(\alpha^*)=1$, then for every $\beta \in \Sigma_0^T$, at least a $\delta/48$ fraction of the constraints of $I$ are falsified by $(\alpha,\beta)$.
    \end{itemize}
\end{lemma}
As it will be pertinent to our analysis, we recall the construction of the \Call{AssignmentTester}{}, run on a Boolean circuit $C$ with input variables $X$.  

\paragraph{Variable Introduction.} We construct the sets of additional Boolean variables $Y$ and $Z$ as follows:
\begin{enumerate}
    \item 

    Arithmetize the circuit $C$ as a set of quadratic equations ${\cal P}$ such that $C(x) = 1$ if and only if $P(x) = 1$ for every $P \in \mathcal P$. To do so, we introduce, for each gate $g$ in $C$, a new Boolean variable $y_g$. Let $Y$ be the collection of all $y_g$ variables introduced. As well, we add the following quadratic constraints to ${\cal P}$: for every gate $g$ in $C$ with children $g_1,g_2$, we add one of the following constraints: 
    \begin{itemize}
        \item If $g= \wedge$: $y_{g_1} y_{g_2}-y_g = 0$ if $g_1,g_2$ are the children to $g$.
        \item If $g = \lor$: $y_{g_1} + y_{g_2}+ y_{g_1} y_{g_2}-y_g = 0$ if $g_1,g_2$ are the children to $g$.
        \item If $g = \neg$: $y_{g_1} +y_g -1 = 0$ if $g_1$ is the child to $g$.
    \end{itemize}
    For the remaining pairs $(g_i,g_j)$ and triples $(g_i,g_j,g_k)$, with $i \geq j,k$, for which we have not introduced a constraint, we add a trivial constraint which is always satisfied.
    Let $k~:=|X|+|Y| = |X|+|C|$, where $|C|$ denotes the number of gates in the circuit $C$.

    \item 
    Introduce sets of variables $L \in \{0,1\}^{k}$ and $Q \in \{0,1\}^{k^2}$, and let $Z=L \cup Q$. These variables will purportedly represent the Hadamard encoding $L$ of an assignment $\alpha \in \{0,1\}^{X \cup Y}$ and the Hadamard encoding of $\alpha \otimes \alpha$. 
    We will think of $L$ as a table indexed by strings $s \in \{0,1\}^k$ such that $L(s) = \sum_{i=1}^k s_i \alpha_i $, and similarly for $Q$; that is, $Q(t) = \sum_{1 \leq i,j \leq k} t_{ij}\alpha_i \alpha_j$ for $t \in \{0,1\}^{k^2}$. 
\end{enumerate}
In total, $Y \cup Z$ contains $O(k)+2^{k(k+1)}$-many Boolean variables.

\paragraph{Constraint Introduction.} The purpose of the set of constraints that will be introduced is to verify that an assignment to the variables $\{0,1\}^{X \cup Y \cup Z}$ encodes a satisfying assignment to the circuit $C$. 
We will describe the constraints using a probabilistic language, where each edge $e$ that is introduced will have an associated weight. 
Furthermore, the edges will initially be hyperedges (involving up to $6$ variables); they will be reduced to regular edges in the \emph{sparsification} step below.
We will construct a \emph{weighted} instance $I~:=~(X \cup Y \cup Z, E, \{0,1\}, {\cal R})$, by populating the edge set $E$ and constraint set ${\cal R}$. 
Every time we ``introduce a constraint'' we add an edge to $E$ and a corresponding constraint $R_e$ to ${\cal R}$. 
We will then normalize the instance by adding copies of each edge $e \in E$ proportional to its weight to obtain the final instance. 
\begin{enumerate}
    \item Verify that $L$ is a Hadamard code (of some assignment in $\{0,1\}^k$): Sample $x,y \sim \{0,1\}^{k}$ and check that $L(x)+L(y)=L(x+y)$; that is, we are running the BLR Linearity test~\cite{BlumLR93} on $L$. That is, for every $x,y \in \{0,1\}^k$ we introduce a constraint which accepts iff $L(x)+L(y)=L(x+y)$, with associated weight $w_1~:=~2^{-2k}$. 

    \item Verify that $Q$ is a Hadamard code (of some assignment in $\{0,1\}^{k^2}$): For every $x,y \in \{0,1\}^{k^2}$, introduce a constraint which accepts iff $Q(x)+Q(y)=Q(x+y)$, with associated weight $w_2~:=~2^{-2k^2}$.
    \item Verify that $L$ and $Q$ encode the same assignment: Let $\Call{SelfCorrect}{L,x}$ be the procedure which samples $x' \sim \{0,1\}^k$ and outputs $L(x+x')-L(x')$. In this step we sample $x,y \sim \{0,1\}^k$ and check whether 
    \[ \Call{SelfCorrect}{L,x} \cdot \Call{SelfCorrect}{L,y} = \Call{SelfCorrect}{Q, x \otimes y}.\]
    That is, for every $x,y,x',y' \in \{0,1\}^k$ and $q \in \{0,1\}^{k^2}$, we introduce the constraint
    \[ \big(L(x+x')-L(x') \big)\big(L(y+y')-L(y') \big) = Q(x \otimes y +q) - Q(q),\]
    with weight $w_3~:=~2^{-(4k+k^2)}$. 

    \item Verify that the assignment that is referred to by $L$ and $Q$ satisfies the circuit $C$: Sample $r \sim \{0,1\}^{|{\cal P}|}$ and let $s_0 \in \{0,1\}, s \in \{0,1\}^{k}, t \in \{0,1\}^{k^2}$ be such that 
    \[ 
        \sum_{P \in {\cal P}} r_P P(x) = s_0 + \sum_{i=1}^k s_i x_i + \sum_{1 \leq i, j \leq k} t_{ij} x_ix_j, 
    \]
    and check whether 
    $s_0 + \Call{SelfCorrect}{L,s}+\Call{SelfCorrect}{Q,t}=0$. That is, for every $r \in \{0,1\}^{|{\cal P}|}, x \in \{0,1\}^k, q \in \{0,1\}^{k^2}$ introduce the constraint
    \[ 
        s_0 + L(s+x)-L(x) + Q(t+q)-Q(q) = 0,
    \]
    with weight $w_4~:=~ 2^{-(|C|+k+k^2)}$, noting that $|{\cal P}|=|C|$. 

    \item Verify that the assignment referred to by $L$ and $Q$ is the same assignment as encoded by the variables $X$. Let $e_i$ denote the $i^{th}$ standard basis vector.
    Sample $i \sim [|X|]$ and check whether 
    \[
        \Call{SelfCorrect}{L,e_i}=X_i.
    \] 
    That is, for every $x \in \{0,1\}^k$, and every $1 \leq i \leq |X|$, introduce $L(e_i+x)-L(x) = X_i$.
    Weight each of these constraints by $w_5~:=~2^{-k}/|X|$. 
\end{enumerate}
We convert this instance $I$ to an unweighted instance by replacing each constraint edge pair $(e,R_e)$ of type (i), for $i \in [5]$, with $w_i N$-many copies of it, where $N$ is the least common multiple of $1/w_1,\ldots,1/w_5$. 
Observe that each of these constraints depend on at most $6$ variables in $X \cup Y \cup Z$. Finally, as we would like to measure swap sensitivity (rather than sensitivity) we would like to be able to recover any instance obtained by swapping $C$ for another circuit $C$ of the same size (encoding a constraint in our instance).\footnote{When we apply the Assignment Tester to the constraints of our instance, we will assume that every constraint is encoded by a circuit of the same size, by taking the maximum.} To accommodate this, let $\alpha$ be the maximum number of times that a specific $L(s)$ or $Q(q)$ appears in the constraints obtained from step (4). For every $s,x \in \{0,1\}^k$ and $t,q \in \{0,1\}^{k^2}$ for which a constraint of type (4) does not exist, we introduce a trivial (always satisfied) constraint on the variables $L(s+x),L(x),Q(t+q),Q(q)$, with multiplicity $\alpha w_4N$.

It remains to reduce these to binary constraints.

\paragraph{Sparsification.} Let $I=(X \cup Y \cup Z,E,\Sigma_0, {\cal R})$ be the instance constructed so far. We reduce these $6$-ary constraints to binary constraints by introducing an additional set of $2^6$-ary variables $W~:=~\{w_e : e \in E\}$. 
We replace each pair of edge and constraint $(e, R_e)$ belonging to $E$ and ${\cal R}$ as follows: Let $v_1,\ldots, v_t$ for $2 \leq t \leq 6$ be the set of variables on which $R_e$ depends.
Replace $R_e$ by constraints $R_{e,i}$ on $(w_e, v_i)$ for $i \in [t]$, where $(a,b) \in R_{e,i}$ if $b$ is consistent with $a$ and $a$ satisfies $R_e$. 
Let $I = (X \cup T, E, \Sigma_0, {\cal R})$, where $T :=X \cup Y \cup Z \cup W$, be the resulting instance that we have constructed. This completes our description of the \Call{AssignmentTester}{}.\\ 

We record some useful observations.
\begin{observation}
\label{obs:ATProps}
    Let $C$ be any circuit on variables $X$ and consider $I=\Call{AssignmentTester}{C}$. 
    Then following hold: 
    \begin{enumerate}
    \item The parameters $N$ and $k$ depend only on the number of gates of $C$.
    \item Replacing $C$ by any other circuit of the same size on the input variables $X$ will only modify $O(N)$-many constraints from (4). 
    \item Each $X$-variable has degree $D:=3N/|X|$, as each constraint in step (5) depends on exactly three variables and by the sparsification step we replace each such constraint by three binary constraints. 
\end{enumerate}

\end{observation}

\subsubsection{Reduction Procedure}
Now, we describe the procedure \Call{AlphabetReduction}{}, building on \Call{AssignmentTester}{}.
Let $I = (V,E, \Sigma, {\cal R}) \in {\cal I}$ be a binary CSP instance with $n$ variables.
First, we encode the elements of $\Sigma$ in binary by applying the Hadamard code $\mathsf{Had} : \Sigma \rightarrow \{0,1\}^\ell$, where $\ell~:=~4\log |\Sigma|$; 
hence each vertex $v \in V$ is encoded by a set $X_v$ of $\ell$-many Boolean variables.
For each edge $e =(u,v)\in E$ we will replace the constraint $R_e \subseteq \Sigma^2$ by a Boolean circuit $C_e:\{0,1\}^{2 \ell} \rightarrow \{0,1\}$ such that $C_e(\alpha,\beta) = 1$ iff there are $a,b \in \Sigma$ such that $\mathsf{Had}(a)=\alpha, \mathsf{Had}(b)= \beta$, and $(a,b) \in R_e$. 
Let $X_e = X_u \cup X_v$ be the set of Boolean variables (encoding the two endpoints of $e$) on which $C_e$ depends. 
Finally, we will assume without loss of generality that the size of the circuits $C_e$ is the same for all $e \in E$ by padding. In particular, by \Cref{obs:ATProps}, the parameters $N$ and $k$ will be the same for every $e \in E$.

To generate the instance $I':=\Call{AlphabetReduction}{I}$, we will feed the circuit $C_e$ that simulates each constraint $R_e$ to the \Call{AssignmentTester}{}. 
Let $I_e = (X_e \cup T_e, E_e, \Sigma_0, {\cal R}_e):=\Call{AssignmentTester}{C_e}$. 
Define the instance $I':=(X' \cup T',E', \Sigma' = \Sigma_0, {\cal R}')$ as $X'=\bigcup_{e \in E} X_e$, $T'= \bigcup_{e \in E} T_e$, $E' := \bigcup_{e \in E} E_e$, and ${\cal R}' := \bigcup_{e \in E} {\cal R}_e$.

The goal of this section is to show the following:
\begin{lemma}\label{lem:alphabet-reduction}
    Let $\mathcal I$ be a swap-closed family of CSP instances and let $\mathcal I' = \Call{AlphabetReduction}{\mathcal I}$.
    Then we have 
    \[
       \swapsens_{1,1-\varepsilon'}(\swapclo(\mathcal I')) \geq \frac{  \swapsens_{1,1-\varepsilon}(\mathcal I)}{C_I \cdot C_\sigma},
    \] 
    where $C_I = N$, $C_\sigma = O(1/\log |\Sigma|)$, 
    and $\varepsilon' \geq \varepsilon/C$ for some universal constant $C>0$, and parameter $N$ from the \Call{AssignmentTester}{}. 
\end{lemma}

Consider the following algorithm $T_\sigma$, which given an assignment $\sigma' : V' \to \Sigma'$ to $I'$, returns an assignment $\bm \sigma : V \to \Sigma$ to $I$ as follows. 
For each $u \in V$, denote by $\sigma'[u]$ the restriction of $ \sigma'$ to the block of variables which represent $ u$'s label. Then $ \sigma'[u]$ is a purported Hadamard codeword on $\ell = 4 \log |\Sigma|$-many bits. 
Let $c_u$ be the closest codeword to $\sigma'[u]$.
Consider the ``Swiss cheese'' $S \subseteq \{0,1\}^\ell$ obtained from the hypercube $\{0,1\}^\ell$ by removing any point which is within the unique decoding radius (the Hamming ball of radius $\ell/4$) of any codeword; an illustration is given in \Cref{fig:swiss}.
Let $\delta_u = \Ham(\sigma'[u],c_u)/\ell$.
Note that $\delta_u \in [0,1/4]$.
Then, we choose a threshold $\tau \in [0,1]$ uniformly at random (we use this threshold for all $u \in V$).
Then, we set $\bm \sigma(u) = a$ if $4\delta_u \leq \tau$, where $a \in \Sigma$ is such that $\mathsf{Had}(a) = c_u$.
Otherwise, we set $\bm \sigma(u) = r$, where $r$ is an arbitrary but fixed 
label in $\Sigma$.
Let $\mu_u$ be this distribution, where $\mu_u(1)$ denotes the first event and $\mu_u(2)$ denotes the second.

\begin{figure}
    \centering
    \begin{tikzpicture}

        \filldraw[color=black!60, fill=betterYellow!16, very thick](0,-1) -- (2,-3) -- (6,-2) -- (6,1.5) -- (4,3.5) -- (0,2.5);

        \draw[color=black!60, very thick] (4,0)-- (4,3.5) -- (6,1.5) -- (6,-2) -- cycle;

        \draw[color=black!60, very thick] (0,-1)-- (0,2.5) -- (2,0.5) -- (2,-3) -- cycle;

        \draw[color=black!60,  very thick] (0,-1) -- (4,0);

        \draw[color=black!60, very thick] (2,0.5) -- (6,1.5);

        \filldraw[color=black!60, fill=white, very thick](0,2.5) circle (1.5);

        \filldraw[color=white, fill=white, very thick](-0.03,1) -- (1.45,2.89) -- (1.43,4) -- (-1.5,4) --  (-1.5,1) --cycle;

        \draw[color=black!60, very thick, dashed] (0,2.5) ellipse (1.3cm and 0.5cm);

        \node at (0,1.8)(tau1){$\tau$};
        \filldraw[color=black!60,fill=red!12, very thick](0,2.5) circle (0.2);

        \filldraw[color=black!60, fill=white, very thick](6,-2)  circle (1.5);

        \filldraw[color=white, fill=white, very thick] (4.55,-2.38) -- (6,-0.5) -- (7.5,-0.5) -- (7.5, -3.5) -- (4.55,-3.5) --  cycle;

        \node at (6,-1.3)(tau2){$\tau$};
        \draw[color=black!60, very thick, dashed] (6,-2) ellipse (1.3cm and 0.5cm);
        \filldraw[color=black!60,fill=red!12, very thick](6,-2) circle (0.2);
    \end{tikzpicture}
    \caption{Conceptual illustration of the ``Swiss cheese'' $S \subseteq \{0,1\}^\ell$. Vertices of $\{0,1\}^\ell$ which are codewords are indicated by red circles, and the Hamming balls of radius $\ell/4$ around each one of them has been removed from $\{0,1\}^\ell$ to form $S$. The random threshold $\tau$ chosen in the recovery procedure is indicated by the dashed circle.}
    \label{fig:swiss}
\end{figure}
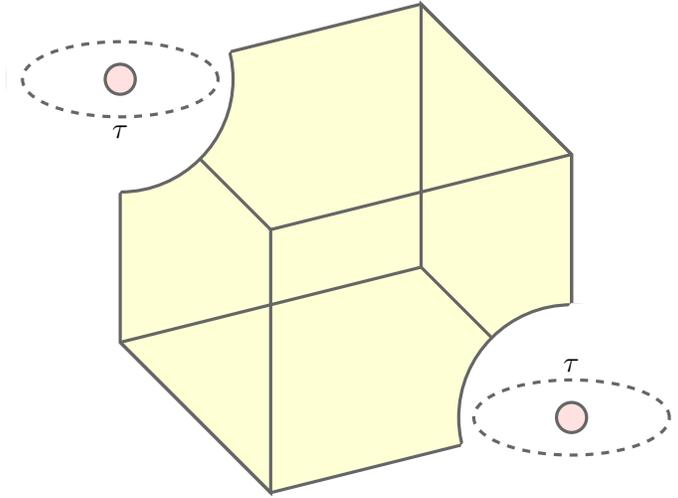

We claim that $(\Call{AlphabetReduction}{}, T_\sigma)$ is a $(1,1-\varepsilon, c'= 1, s' = 1-\varepsilon', C_I=N, C_\sigma = 8/\ell)$-sensitivity-preserving reduction, where $\varepsilon' = \Theta(\varepsilon)$.
The analysis for $c'$ is obvious.
Also, $C_I = N$ suffices by \Cref{obs:ATProps}. 
Hence, we focus on analyzing $C_\sigma$ and $s'$ in the rest of this section.

\subsubsection{Sensitivity Increase in the Recovery Procedure}
We now analyze the sensitivity increase in the recovery procedure $T_\sigma$.
\begin{lemma}
     The choice $C_\sigma = 8/\ell$ satisfies \Cref{item:reduction-assignment-distance} of \Cref{def:sensitivity-reduction}.
\end{lemma}    
\begin{proof}
    Consider any pair of assignments $\sigma', \tilde \sigma' : V' \to \Sigma'$ such that $\Ham(\sigma', \tilde \sigma')=1$. 
    As well, let $\bm \sigma = T_\sigma(\sigma')$ and $\tilde { \bm \sigma} = T_\sigma(\tilde \sigma' )$. 
    Then, 
    \begin{align*}
         \EMD(\bm \sigma,\tilde{\bm \sigma}) 
        & \leq \sum_{u \in V} \TV \Big(\bm \sigma(u) ,~ \tilde{\bm \sigma}(u)  \Big).
    \end{align*}
    Suppose $\Ham(\sigma'[u],\tilde \sigma'[u]) = 1$ for a vertex $u \in V$.
    Let $c_u, \tilde c_u \in \{0,1\}^\ell$ be the closest codewords to $\sigma'[u]$ and $\tilde \sigma'[u]$, respectively.  we consider two cases based on whether these codewords are the same. 

    \begin{description}
    \item[Case 1. ] If $c_u = \tilde c_u$: Then, 
    \begin{align*}
        & \TV \qty(\bm \sigma(u),~ \tilde{\bm \sigma}(u) ) 
        = \TV \qty(\mu_u, \tilde \mu_u ) 
        = \frac{1}{2} \Big( |\mu_u(1) - \tilde \mu_u(1) | + | \mu_u(2) - \tilde \mu_u(2)| \Big)  \\
        &= \frac{1}{2 \cdot (\ell/4)} \Big(|\Ham(\sigma'[u], c_u) - \Ham(\tilde \sigma'[u], c_u)| + |\Ham(\sigma'[u], S) - \Ham(\tilde \sigma'[u], S)| \Big) \\
        &\leq \frac{1}{\ell/4} \cdot \Ham(\sigma'[u], \tilde \sigma'[u]). \tag{By the triangle inequality} \\
        & = \frac{4}{\ell}.
    \end{align*}
    \item[Case 2.] If $c_u \neq \tilde c_u$: As $\Ham(\sigma', \tilde \sigma')=1$, this can only occur if $\sigma'[u], \tilde \sigma'[u]$ lie on the unique decoding radius of their respective codewords.
    Let $a,b \in \Sigma$ be such that $\mathsf{Had}(a) = c_u, \mathsf{Had}(b)=\tilde c_u$.
    Then $\bm \sigma(u), \tilde{\bm \sigma}(u)$ take on values $a,b$, respectively, with probability at most $1/(\ell/4)$.
    Then, we have
    \begin{align*}
        & \TV \qty(\bm \sigma(u),~ \tilde{\bm \sigma}(u) ) 
        \leq 
        \qty|\Pr \big[\bm \sigma(u) = a \big] - \Pr \big[ \tilde {\bm \sigma}(u) = b \big]| 
        \leq 
        \frac{2}{\ell/4} = \frac{8}{\ell}.
    \end{align*}
    \end{description}
    Hence, the claim follows.
\end{proof}

\subsubsection{Gap Analysis}
Now we analyze the gap decrease of \Call{AlphabetReduction}{}.
\begin{lemma}
    The choice $s' = 1-\varepsilon'$ for  $\varepsilon' = \varepsilon/C$ satisfies \Cref{item:reduction-soundness} of \Cref{def:sensitivity-reduction}, where $C>0$ is a universal constant.
\end{lemma}
\begin{proof}

     Let $\sigma'$ be an assignment to $I'$ with $\val_{I'}(\sigma') \geq 1-\varepsilon'$.
    Let $\bm \sigma$ be the assignment recovered from $\bm \sigma'$ by $T_\sigma$.
    Our goal is to show that $\E \val_I(\bm \sigma) \geq 1-\varepsilon$.
    We can handle the situation that we have a random assignment $\bm \sigma'$ for $I'$ by the conditioning argument as in the proof of \Cref{lem:degree-reduction}.

    (Generalizing this argument to the case where $\sigma'$ is a random assignment is straightforward.)

    For $e \in E$, let $\varepsilon_e$ be the probability that $\bm \sigma$ violates the constraint on $e$. 
    Also, let $\varepsilon'_e$ be the fraction of the constraints in $I_e$ violated by $\sigma'$.
    Note that $\sum_{e \in E}\varepsilon'_e /m \leq \varepsilon'$.
    For $u \in V$, let $c_u \in \{0,1\}^\ell$ be the closest codeword to $\sigma'[u]$.
    Let $a_u \in \Sigma$ be such that $\mathsf{Had}(a_u) = c_u$.
    Let $\delta_u = \Ham(\sigma'[u],c_u)/\ell$ .
    Then, the probability that $\bm \sigma(u) = a_u$ is 
    \begin{align*}
        & p_u := \frac{\ell/4-\Ham(\sigma'[u], c_u)}{\ell/4}
        = \frac{1/4-\delta_u}{1/4} = 1 - 4\delta_u.
    \end{align*}

    Fix an edge $e = (u,v) \in E$ of $I$.
    Let $(c^*_u,c^*_v)$ be the satisfying assignment to $I_e$ which is closest to $(\sigma'[u],\sigma'[v])$.     
    As $(c^*_u,c^*_v)$ is satisfying, they must be Hadamard codewords, and hence there are $a^*_u,a^*_v \in \Sigma$ such that $\mathsf{Had}(a^*_u)=c^*_u$, $\mathsf{Had}(a^*_v)=c^*_v$, and $(a^*_u,a^*_v) \in R_e$.       
    Let $p_e$ be the probability that both $\bm \sigma(u) = a_u$ and $\bm \sigma(v) = a_v$ hold.
    Because we use threshold rounding, we have $p_e = 1 - \max\{4\delta_u,4\delta_v\} = \min\{p_u,p_v\}$.

    We consider two cases.
    \begin{description}
    \item[Case 1.] $(a_u,a_v) \not \in R_e$.
        Then at least one of $a^*_u \neq a_u$ or $a^*_v \neq a_v$.
        Suppose without loss of generality that the first holds.
        Then, 
        \begin{align*} 
            \frac{\ell}{2} &\leq \Ham \Big(\mathsf{Had}(a^*_u),~ \mathsf{Had}(a_u) \Big) \tag{As Hadamard codewords have distance $1/2$. } \\
            &= \Ham \Big(c^*_u,~ c_u \Big) \tag{$\mathsf{Had}(a_u) = c_u$}\\
            &\leq \Ham \Big(c^*_u, \sigma'[u] \Big) + \Ham \Big(\sigma'[u], c_u\Big) \\
            &\leq 2 \Ham \Big( c^*_u, \sigma'[u] \Big),
        \end{align*}
        where the final inequality follows as $c_u$ is the closest codeword to $\sigma'[u]$. 
        Hence, $(\sigma'[u], \sigma'[v])$ is  $(1/8)$-far from $(c^*_u,c^*_v)$, and hence $(1/(8\cdot48))$-far from any satisfying assignment to $I_e$ by \Cref{lem:assignTest}.        
        This means $\varepsilon'_e \geq 1/(8\cdot48) \geq \varepsilon_e / 384$, where the last inequality is by $\varepsilon_e \leq 1$. 
    \item[Case 2.] $(a_u,a_v) \in R_e$.
        Suppose without loss of generality that $\delta_u \geq \delta_v$ holds.
        Then, $\Ham(\sigma'[u], c^*_u) \geq \Ham(\sigma'[u], c_u) \geq \delta_u \ell$ holds because $c_u$ is the closest codeword to $\sigma'[u]$.
        Hence, $\sigma'([u],[v])$ is at least $(\delta_u/2)$-far from $(c^*_u,c^*_v)$, and it follows that $(\delta_u/2)$-far from any satisfying assignment to $I_e$ by \Cref{lem:assignTest}.
        This implies $\varepsilon'_e \geq \delta_u/(2\cdot 48)$.

        On the other hand, the probability that we have both $\bm \sigma(u) = a_u$ and $\bm \sigma(v) = a_v$ is $p_u$.
        Hence $\epsilon_e \leq 1 - p_u = 4\delta_u$ and we have  $\varepsilon'_e \geq \varepsilon_e/(2 \cdot 48 \cdot 4) = \varepsilon_e/384$.
    \end{description}
    Combining the two cases, we have
    \[
        \varepsilon' = \frac{1}{m}\sum_{e \in E} \varepsilon'_e
        \geq \frac{1}{m} \sum_{e \in E} \frac{\varepsilon_e}{384}
        = \frac{\cost_I(\bm \sigma)}{384}.
    \]
    Hence, we have $\cost_I(\bm \sigma) \leq \epsilon$ by setting $\varepsilon' = \varepsilon/384$.
\end{proof}

\subsubsection{Sensitivity Recovery}
Finally, we combine the \Call{AlphabetReduction}{} procedure with \Call{DegreeReduction}{} in order to restore the our sensitivity bound.

\begin{lemma}\label{lem:alphabet-reduction+degree-reduction}
    Let $\mathcal I$ be a swap-closed family of instances of a binary CSP and consider the family of instances $\mathcal I' = \Call{DegreeReduction}{\Call{AlphabetReduction}{\mathcal I}}$.
    Then for any $\varepsilon>0$, we have 
    \[
        \swapsens_{1,1-\varepsilon'}(\swapclo(\mathcal I')) \geq  \Omega\qty(\swapsens_{1,1-\varepsilon}(\mathcal I) ),
    \] 
    for $\varepsilon'=\varepsilon/C$ for some universal constant $C>0$.
\end{lemma}
\begin{proof}
    Let $I'=(X' \cup T',E', \Sigma' = \Sigma_0, {\cal R}')= \Call{AlphabetReduction}{I}$ for an instance $I \in {\cal I}$, where $X'$ are the vertices which represent the Hadamard encoding of the vertices of $I$. Let $D$ be the maximum degree of any $X'$-vertex among all $I \in {\cal I}'$. Our aim is to apply \Cref{cor:degree-reduction} to $I'$. Define the marked CSP $\hat I'$ which is $I'$ with marked set of vertices $S=X'$. 
    Then by \Cref{lem:alphabet-reduction} and \Cref{cor:degree-reduction}, we have
    \begin{align*}
        & \swapsens_{1,1-\varepsilon'}(\swapclo(\mathcal I'))
        \geq 
        \frac{D \cdot \swapsens_{1,1-\varepsilon}(\mathcal I)}{C_I C_\sigma} \\
        & = 
        \Omega\qty(\frac{N/\log|\Sigma| \cdot \swapsens_{1,1-\varepsilon}(\mathcal I)}{N \cdot (1/\log |\Sigma|)})   
        = \Omega\qty(\swapsens_{1,1-\varepsilon}(\mathcal I)).
        \qedhere
    \end{align*}
\end{proof}

\subsection{Proof of \Cref{thm:label-cover}}\label{subsec:label-cover}

\begin{proof}[Proof of \Cref{thm:label-cover}]
    Starting with the $\Omega(1/n)$ gap between completeness and soundness from \Cref{lem:2-xor}, we gradually increase this gap  by applying the reductions developed in this section.

    Each round consists of \Cref{lem:expanderization}, \Cref{lem:powering+degree-reduction} with a sufficiently large constant $t$, and \Cref{lem:alphabet-reduction+degree-reduction}, in this order, and we apply the reductions for $r := O(\log_t n)$ rounds, where $\delta>0$ is a sufficiently small constant.
    The gap amplification and subsequent degree reduction steps increase the number of vertices by a factor of $O(d^t)$, where $d$ is the degree of the instance at the beginning of the round, and the alphabet reduction and subsequent degree reduction steps increase the number of vertices by a factor of $2^{O(k^2)}$, where $k$ is from \Call{AssignmentTester}{}.
    The instance size after $r$ rounds is 
    \[
        N := n \cdot \qty(O\qty(d^t) \cdot 2^{O(k^2)})^r = n^{O(t/\log t)}.
    \]
    This implies $n = N^{\Omega(\log t/t)}$.
    Each round decreases the sensitivity lower bound by a constant.
    Then after $r$ rounds, the sensitivity lower bound becomes
    \[
        \frac{n}{O(1)^r} = n^{1-O(1/\log t)} = N^{\Omega(\log t/t)}.
    \]
    Each round increases the gap between the completeness and soundness by $t/C$ for some universal constant $C>0$.
    Hence by choosing the hidden constant in $r$ to be large enough, the gap after $r$ rounds is 
    \[
        \Omega\qty(\frac{1}{n}) \cdot \qty(\frac{t}{C})^r = 
        \Omega(1).
    \]
    To obtain label cover instances, observe that the sparsification process in the alphabet reduction step generates a label cover instance. Therefore, we can obtain label cover instances simply by omitting the degree reduction step at the end.
    Finally, \Cref{lem:sens-and-swap-sens} yields the theorem.
\end{proof}

\section{Maximum Clique and Related Problems}\label{sec:max-clique}

A vertex set $S \subseteq V$ in a graph $G = (V, E)$ is called a \emph{clique} if every pair of vertices in $S$ is adjacent in $G$. 
In the \emph{maximum clique problem}, given a graph $G = (V, E)$, the goal is to find a clique of maximum size.
We first show a sensitivity lower bound for $(1-\varepsilon)$-approximation algorithms for some $\varepsilon>0$ in \Cref{subsec:max-clique-1-eps}.
As a corollary, we also obtain a lower bound for $(1+\varepsilon)$-approximation algorithm for the minimum vertex cover problem.
Then, we show a sensitivity lower bound for $n^{-\varepsilon}$-approximation algorithm for some $\varepsilon>0$ in \Cref{subsec:max-clique-eps}.
For a graph family $\mathcal G$, let $\clo(\mathcal G)$ denote the \emph{deletion closure} of $\mathcal G$, i.e., the family of graphs obtained from a graph in $\mathcal G$ by deleting edges.

\subsection{Lower Bounds for \texorpdfstring{$(1-\varepsilon)$}{1-eps}-Approximation Algorithms}\label{subsec:max-clique-1-eps}

In this section, we show the following:
\begin{theorem}\label{thm:max-clique-1-eps}
    There exist universal constants $\varepsilon,\delta>0$ such that any $(1-\varepsilon)$-approximation algorithm for the maximum clique problem has sensitivity $\Omega(n^{\delta})$.
\end{theorem}

For $1 \geq c \geq s \geq 0$, we define $\mathsf{MaxClique}_{c,s}$ as the problem that, given a graph $G=(V,E)$ on $n$ vertices with a clique of size at least $cn$, the goal is to find a clique of size at least $sn$.
We use the standard reduction from $\mathsf{LabelCover}$ to $\mathsf{MaxClique}$~\cite{feige1996interactive}, which we call \Call{FGLSS}{}.
Let $I = (U, V, E, \Sigma_U,\Sigma_V,\mathcal{R} = \{R_e\}_{e \in E})$ be a satisfiable label cover instance.
We construct a graph $G' = (V',E')$, where $V'= \{v_{e,a,b} : e \in E, (a,b) \in R_e\}$.
As for the edges in $G'$, we connect $v_{e_1,a_1,b_1} \in V'$ and $v_{e_2,a_2 ,b_2} \in V'$ by an edge if $(e_1, a_1, b_1)$ and $(e_2, a_2, b_2)$ are consistent.
For each $e \in E$, we define $\mathrm{cloud}(e) = \{v_{e,a,b} : (a,b) \in R_e\}$.
We note that $n' := |V'| = |\Sigma_U| m$.
As well, it is easy to see that $G'$ has a clique of size $m$: For any satisfying assignment $\sigma$ for $I$, the vertex set $\{v_{e,\sigma(u),\sigma(v)} : e = (u,v)\in E\}$ forms a clique in $G'$.

\begin{proof}[Proof of \Cref{thm:max-clique-1-eps}]
    We basically follow the argument in \Cref{sec:reduction}, but we need to slightly modify it because the target problem is not a CSP.

    Let $\varepsilon,\delta>0$ and $d,k\geq 1$ to be the ones given in \Cref{thm:label-cover}.
    Let $\mathcal I$ be the family of satisfiable label cover instances with $n$ variables on finite domains $\Sigma_U,\Sigma_V$ with $|\Sigma_V| \leq |\Sigma_U| = k$, where each variable is incident to at most $d$ constraints.
    Note that $\sens_{1,1-\varepsilon}(\mathcal I) = \Omega(n^\delta)$.

    Let $c := 1/k$ and let $A'$ be an algorithm for $\mathsf{MaxClique}_{c,(1-\varepsilon)c}$.
    Using $A'$, we design an algorithm $A$ for $\mathsf{LabelCover}$.
    Specifically, given a satisfiable instance $I=(U,V,E,\Sigma_U,\Sigma_V,\mathcal{R}) \in \mathcal I$, we first apply \Call{FGLSS}{} to obtain a graph $G'=(V',E')$ of $n'$ vertices.
    Note that $G'$ contains a clique of size $|V'|/|\Sigma_U| = cn'$.
    Then, we apply the algorithm $A'$ to obtain a clique $\bm C$ of size at least $(1-\varepsilon)cn'$.

    Next, we describe a procedure $T_\sigma$ that recovers an assignment $\bm \sigma$ for $I$ from a clique $C$ for $G'$.
    Let $U(C) = \{u \in U : v_{uv,a,b} \in C\}$ and $V(C) = \{v \in V: v_{uv,a,b} \in C\}$.
    It is easy to observe that for each $u \in U(C)$, there is a value $a_u \in \Sigma_U$ such that, if $C$ has a vertex of the form $v_{e,a,b}$ with $u$ being an endpoint of $e$, then $a = a_u$.
    Thus, we can define $\sigma(u) = a_u$. 
    For a variable $u \in U \setminus U(C)$, we set $\sigma(u)$ to be an arbitrary but fixed value in $\Sigma_U$.
    Similarly, for each $v \in V(C)$, there is a value $b_v \in \Sigma_V$ such that, if $C$ has a vertex of the form $v_{e,a,b}$ with $v$ being an endpoint of $e$, then $b = b_v$.
    Thus, we can define $\sigma(v) = b_v$.
    For a variable $v \in V \setminus V(C)$, we set $\sigma(v)$ to be an arbitrary but fixed value in $\Sigma_V$. 

    Now, we show that the expected value of $\bm \sigma := T_\sigma(\bm C)$ is at least $1-\varepsilon$.
    By the observation that $\bm C$ may contain at most a single vertex from the cloud of each $e \in E$, it follows that there are more than $(1 - \varepsilon)cn' = (1-\varepsilon)m$ clouds in expectation from which $\bm C$ contains a vertex. 
    Let $\bm F \subseteq E$ be the set of edges $e \in E$ such that $\bm C$ contains some vertex from the cloud of $e$.
    We argue that $\bm \sigma$ satisfies all edges in $\bm F$, hence they satisfy at least $|\bm F|$ constraints, which is a $ (1 - \varepsilon)$-fraction of the constraints in expectation.
    Indeed, if $e=(u,v) \in \bm F$ then there is a vertex of the form $v_{e,a,b}$ in $\bm C$.
    Then, $(a,b)$ satisfies the constraint $R_e$ by construction, and we have $\sigma(u) = a$ and $\sigma(v) = b$ by the choice of $\bm \sigma$.    

    Next, we discuss sensitivity.
    Let $I,\tilde I \in \mathcal I$ be two instances, where $\tilde I$ is obtained from $I$ by deleting a constraints.
    Let $G',\tilde G'$ denote the graphs obtained from $I, \tilde I$, respectively, by applying \Call{FGLSS}{}.
    Then, it is easy to observe that $|E'\triangle \tilde E'| \leq 2d|\Sigma_U|$.
    Let $\bm C,\tilde{\bm C}$ be the outputs of $A'$ on $G'$ and $\tilde G'$, respectively.
    Then, we have
    \[
        \EMD(\bm C, \tilde{\bm C}) \leq 2dk \cdot \sens(A', \mathcal{G}'),
    \]
    where $\mathcal{G}' = \clo(\Call{FGLSS}{\mathcal I})$.
    Let $\bm \sigma,\tilde{\bm \sigma}$ be the assignments constructed from $\bm C,\tilde{\bm C}$, respectively, by applying $T_\sigma$ above.
    Then, we have 
    \begin{align*}
        & \EMD(\bm \sigma,\tilde{\bm \sigma}) 
        = \EMD(U(\bm C), U(\tilde{\bm C})) + \EMD(V(\bm C),  V(\tilde{\bm C}))  \\
        & \leq 2\EMD(\bm C, \tilde{\bm C})
        = 4dk \cdot \sens(A', \mathcal G'))
    \end{align*}
    As the choice of the algorithm $A'$ was arbitrary, any algorithm for $\mathsf{MaxClique}_{c,(1-\varepsilon)c}$ must have sensitivity $\Omega(n^\delta/dk) = \Omega((n')^\delta)$.
\end{proof}

For a graph $G=(V,E)$, a vertex set $S \subseteq V$ is called a \emph{vertex cover} if every edge in $E$ is incident to some vertex in $S$.
In the \emph{minimum vertex cover problem}, given a graph $G=(V,E)$, the goal is to find a vertex cover of the minimum size.
\begin{corollary}\label{cor:vertex-cover}
    There exist universal constants $\varepsilon,\delta>0$ such that any $(1+\varepsilon)$-approximation algorithm for the minimum vertex cover problem has sensitivity at least $\Omega(n^{\delta})$.
\end{corollary}
\begin{proof}
    First, note that the proof of \Cref{thm:max-clique-1-eps} shows that there exists a universal constants $c >0$ and $\varepsilon>0$ such that any algorithm for $\mathsf{MaxClique}_{c,(1-\varepsilon)c}$ has sensitivity at least $n^{1-o(1)}$.
    Also note that a vertex set $C$ is a clique of a graph $G=(V,E)$ if and only if its complement $\bar C := V \setminus C$ is a vertex cover of the complement graph $\bar G=(V, \binom{V}{2} \setminus E)$.

    Let $A'$ be an arbitrary $(1+\varepsilon c)$-approximation algorithm for the minimum vertex cover problem.
    Then, we consider an algorithm $A$ for the maximum clique problem that, given a graph $G=(V,E)$, first applies $A'$ to $\bar G$ to obtain a vertex cover $S \subseteq V$, and then returns a vertex set $\bar S$.
    When there is a clique of size $cn$, the algorithm $A$ returns a clique of size at least 
    \[
        n-(1+\varepsilon c)(1-c)n
        = - \varepsilon c n + c n + \varepsilon c^2 n
        \geq (1-\varepsilon)cn.
    \]
    Hence, $A$ is a $(1-\varepsilon)$-approximation algorithm and hence its sensitivity must be $\Omega(n^{\delta})$.

    Moreover deleting an edge from $G$ corresponds to adding an edge to $\bar G$, and hence the sensitivity bound applies to $A'$.
\end{proof}

\subsection{Lower Bounds for \texorpdfstring{$n^{-\varepsilon}$}{n^-eps}-Approximation Algorithms}\label{subsec:max-clique-eps}

In this section, we show a sensitivity lower bound for $n^{-\varepsilon}$-approximation algorithms for the maximum clique problem for some constant $\varepsilon>0$.
To this end, we first apply serial repetition to get a lower bound for $\varepsilon$-approximation algorithms for $\mathsf{LabelCover}$ for any $\varepsilon>0$.
\begin{lemma}\label{lem:serial-repetition}
    There exists a universal constant $\delta>0$ such that for any $\varepsilon >0$, any algorithm for $\mathsf{MaxCSP}_{1,\varepsilon}$ on $t$-ary instances of degree at most $O(\log n + \varepsilon^{-1}\log \varepsilon^{-1})$ that succeeds with probability at least $1-O(1/n)$ requires sensitivity $\Omega(n^\delta /(\varepsilon^{-1}\log \varepsilon^{-1}))$, where $t = O(\log \varepsilon^{-1})$.
\end{lemma}
Let $\varepsilon'>0$ be an arbitrary parameter and we set 
\[
    t = O\qty(\frac{\log(1/\varepsilon')}{\varepsilon}) \text{ and } M = \max\qty{\log\frac{3k^n}{\varepsilon} \cdot \frac{(1-\varepsilon/3)^{-t}}{3}, \frac{3m \log(nm)}{t}, \frac{3 m \log n^2}{dt}}.
\]
We consider a randomized transformation, called $\Call{SerialRepetition}{}$, from label cover instances to CSP instances.
Specifically, given an instance $I=(U,V,E,\Sigma_U,\Sigma_V,\mathcal R = \{R_e\}_{e \in E}) \in \mathcal I$, we construct a (random) instance $\bm I' = (U,V,\bm E', \Sigma_U^t ,\Sigma_V^t , \bm{\mathcal R}'=\{R'_e\}_{e \in \bm E'})$ as follows:
For $i = 1,\ldots,M$, sample $t$ edges $\bm e_1,\ldots,\bm e_t \in E$ uniformly at random, and let $\bm e^{(i)}$ be the multiset $\bigcup_{j=1}^t \bm e_j$. Let $R'_{\bm e^{(i)}}$ be the relation obtained by combining $R_{\bm e_1},\ldots, R_{\bm e_t}$, i.e., 
\[
    \bm R'_{\bm e^{(i)}} = \qty{ (a_1,\ldots,a_t,b_1,\ldots,b_t) \in \Sigma_U^t \times \Sigma_V^t : (a_j,b_j) \in R_{\bm e_j}\;\forall j \in [t] }.    
\]
Note that $|\bm E'| = M$ and each relation has arity at most $2t$.

We show several key properties of $\bm I'$.
\begin{claim}\label{cla:serial-repetition-1}
    With probability at least $1-\varepsilon/3$, we have $\val_{\bm I'}(\sigma) < \varepsilon'$ for any assignment $\sigma$ for $I \in \mathcal I$ with $\val_{I}(\sigma) < 1-\varepsilon/3$.  
\end{claim}
\begin{proof}
    Let $\sigma$ be an arbitrary assignment that satisfies at most a $(1-\varepsilon/3)$-fraction of the constraints of $I$.
    Let $\bm X^\sigma_i$ be the indicator random variable that $\sigma$ satisfies the $i^{th}$ constraint of $\bm I'$.
    Then by the Chernoff bound, we have
    \[
        \Pr[ \sum_{i=1}^M \bm X^\sigma_i \geq 2\qty(1-\frac{\varepsilon}{3})^t M ] \leq \exp\qty(-\frac{1}{3}\qty(1-\frac{\varepsilon}{3})^t M) \leq \frac{\varepsilon}{3k^n},
    \]
    by our choice of $M$.
    Finally, by a union bound, with probability at least $1-\varepsilon/3$ we have $\val_{\bm I'}(\sigma) \leq 2(1-\varepsilon/3)^t < \varepsilon'$ for any $\sigma$ with $\val_{I}(\sigma) < 1-\varepsilon/3$.  
\end{proof}
\begin{claim}\label{cla:serial-repetition-2}
    With probability at least $1-1/n$, every edge $e \in E$ is used to construct at most $2tM/m$ constraints in $\bm I'$.
\end{claim}
\begin{proof}
    For an edge $e \in E$, $i \in [M]$, and $j \in [t]$, let $\bm Y^e_{i,j}$ be the indicator that $e$ is used to construct $\bm e^{(i)}$ as the $j$-th edge.
    By Chernoff's bound, we have
    \[
        \Pr\qty[ \sum_{i=1}^M \sum_{j=1}^t \bm Y^e_{i,j} \geq \frac{2tM}{m} ] \leq \exp\qty( -\frac{tM}{3m} ) \leq \frac{1}{nm}.
    \]
    By a union bound, with probability at most $1/n$, every edge $e \in E$ is used to construct at most $2tM/m$ constraints in $\bm I'$.
\end{proof}
\begin{claim}\label{cla:serial-repetition-3}
    With probability at least $1-1/n$, the instance $\bm I'$ has degree at most $dtM/m$. 
\end{claim}
\begin{proof}
    For a variable $v \in U \cup V$, $i \in [M]$, and $j \in [M]$, let $\bm Y^v_{i,j}$ be the indicator that an edge incident to $v$ is used to construct $\bm e^{(i)}$ as the $j$-th edge.
    Note that $\E \bm Y^v_{i,j} \leq d/2m$.
    Hence by Chernoff's bound, we have
    \[
        \Pr\qty[ \sum_{i=1}^M \sum_{j=1}^t \bm Y^v_{i,j} \geq \frac{dtM}{m}  ] \leq \exp\qty( -\frac{dtM}{3m} ) \leq \frac{1}{n^2}.
    \]    
    By a union bound, with probability at most $1/n$, every vertex $v \in U \cup V$ is incident to at most $dtM/m$ constraints in $\bm I'$.
\end{proof}

\begin{proof}[Proof of \Cref{lem:serial-repetition}]
    Let $\varepsilon,\delta>0$ and $d,k\geq 1$ to be the ones given in \Cref{thm:label-cover}.
    Let $\mathcal I$ be the swap-closed family of satisfiable label cover instances with $m$ edges and $n$ variables on finite domains $\Sigma_U,\Sigma_V$ with $|\Sigma_V| \leq |\Sigma_U|=k$, where each variable is incident to at most $d$ other variables.
    Note that we have $\swapsens_{1,1-\varepsilon}(\mathcal I) = \Omega(n^\delta)$ (see the proof of \Cref{thm:label-cover}).

    Let $\mathcal I'$ be the swap closure of instances in the support of $\Call{SerialRepetition}{\mathcal I}$ such that the degree is at most $dtM/m$. 
    Note that $\bm I'$ belongs to $\mathcal I'$ with probability at least $1-1/n$ by \Cref{cla:serial-repetition-3}. 
    Let $A'$ be an arbitrary algorithm for $\mathcal I'$ that attains $\swapsens_{1,\varepsilon'}(\mathcal I')$.
    We design an algorithm $A$ for $\mathcal I$ which operates as follows:
    Given an instance of $I \in \mathcal I$, let $\bm I' = \Call{SerialRepetition}{I}$.
    Then, we apply $A'$ on $\bm I'$ to compute an assignment $\bm \sigma'$, and simply output $\bm \sigma := \bm \sigma'$.

    We first show that $A$ is an algorithm for $\mathsf{LabelCover}_{1,1-\varepsilon}$ on $\mathcal I$.
    Recall that $\val_{I'}(A'(I')) \geq \varepsilon'$ with probability at least $1-O(1/n)$ for any $I' \in \mathcal I'$.
    Let $\bm X_{I'}$ be the indicator of the event that this happens.
    Let $\bm X_{\ref{cla:serial-repetition-1}}$ and $\bm X_{\ref{cla:serial-repetition-3}}$ be the indicators of the events that the events of \Cref{cla:serial-repetition-1} and \Cref{cla:serial-repetition-3} hold, respectively.
    Then, we have
    \begin{align*}
        & \E \val_{I}(A(I))
        \geq 
        \Pr[\bm X_{I'} \wedge \bm X_{\ref{cla:serial-repetition-1}} \wedge \bm X_{\ref{cla:serial-repetition-3}}] \cdot \E [\val_{I}(A(I)) \mid \bm X_{I'} \wedge \bm X_{\ref{cla:serial-repetition-1}} \wedge \bm X_{\ref{cla:serial-repetition-3}}] \\
        & \geq \qty(1-\frac{\varepsilon}{3}-\frac{1}{n}-\frac{1}{n}) \E [\val_{I}(A(I)) \mid \bm X_{I'} \wedge \bm X_{\ref{cla:serial-repetition-1}} \wedge \bm X_{\ref{cla:serial-repetition-3}}] \tag{by \Cref{cla:serial-repetition-1} and \Cref{cla:serial-repetition-3}} \\
        & \geq \qty(1-\frac{\varepsilon}{3}-\frac{1}{n}-\frac{1}{n}) \qty(1 - \frac{\varepsilon}{3}) \tag{by $\val_{I'}(A'(I')) \geq \varepsilon'$ for $I' \in \mathcal I'$} \\
        & \geq 1-\varepsilon.
    \end{align*}

    Next, we analyze the swap sensitivity of $A$.
    Let $I,\tilde I \in \mathcal I$ be two instances with swap distance one.
    We consider a natural coupling between $\bm I'$ and $\tilde{\bm I}'$.
    Note that whenever $\bm I'$ satisfies the claims in \Cref{cla:serial-repetition-2} and \Cref{cla:serial-repetition-3}, so does $\tilde {\bm I}'$.
    Let $\bm X_{\ref{cla:serial-repetition-2}}$ be the indicator of the event that the claim of \Cref{cla:serial-repetition-2} holds.
    Then, the swap sensitivity of $A$ is bounded from above by 
    \begin{align*}
        & \Pr[\bm X_{\ref{cla:serial-repetition-2}} \wedge \bm X_{\ref{cla:serial-repetition-3}}] \cdot \frac{2tM}{m} \cdot \swapsens(A',\mathcal I')
        + \Pr[\bar{\bm X}_{\ref{cla:serial-repetition-2}} \vee \bar{\bm  X}_{\ref{cla:serial-repetition-3}}] \cdot n \\
        & \leq \qty (1 - \frac{2}{n}) \cdot \frac{2tM}{m} \cdot \swapsens(A',\mathcal I') + \frac{2}{n} \cdot n \\
        & \leq O_{\varepsilon,d,k}\qty(\max\qty{\frac{\log(1/\varepsilon')}{\varepsilon'}, \log n} ) \cdot \swapsens_{1,\varepsilon}(\mathcal I').
    \end{align*}
    Hence, we have
    \[
        \swapsens_{1,\varepsilon'}(\mathcal I') 
        = \Omega_{\varepsilon,d,k}\qty(\frac{\swapsens_{1,1-\varepsilon}(\mathcal I)}{\max\qty{\frac{\log(1/\varepsilon')}{\varepsilon'}, \log n} }) 
        = \Omega_{\varepsilon,d,k}\qty(\frac{n^\delta}{\log n + \frac{\log(1/\varepsilon')}{\varepsilon'} }) 
        \qedhere
    \]
    We obtain the claimed lower bound for swap sensitivity by replacing $\varepsilon'$ with $\varepsilon$ and slightly decreasing $\delta$.
    Finally, \Cref{lem:sens-and-swap-sens} gives a lower bound on sensitivity.
\end{proof}

\begin{theorem}\label{thm:max-clique-eps}
    There exist universal constants $\varepsilon,\delta >0$ such that any algorithm for the maximum clique problem that outputs an $n^{-\varepsilon}$-approximate clique with probability $1-O(1/n)$ has sensitivity $\Omega(n^{\delta})$.
\end{theorem}
\begin{proof}
    Let $\varepsilon > 0$ be a parameter which will be determined later.
    Let $\delta>0$ and $t \geq 1$ be the ones given in \Cref{lem:serial-repetition} with $\varepsilon$ in the statement being replaced with $n^{-\varepsilon}$; in particular, $t = O(\varepsilon \log n)$.
    Let $\mathcal I$ be the set of satisfiable label cover instances with at most $n$ variables, each relation having arity at most $t$, and each variable being incident to at most $d = O(n^\varepsilon \log n^\varepsilon)$ constraints, given by \Cref{lem:serial-repetition}.
    Note that any algorithm for $\mathcal I$ that outputs an $n^{-\varepsilon}$-approximate assignment with probability $1-O(1/n)$ has sensitivity $\Omega(n^\delta)$.

    We consider a slight modification of the FGLSS reduction described in \Cref{subsec:max-clique-1-eps}.
    Given a CSP instance $I=(V,E,\Sigma,\mathcal{R} = \{R_e\}_{e \in E}) \in \mathcal I$, we construct a graph $G'=(V',E')$, where $V' = E \times \Sigma^t$, as follows:
    For each hyperedge $e \in E$ and every assignment $\alpha \in \Sigma^t$, add a node $v_{e,\alpha}$ to $V'$.
    Then, we connect $v_{e_1,\alpha_1}$ and $v_{e_2,\alpha_2}$ if $\alpha_1 \in R_{e_1}$, $\alpha_2 \in R_{e_2}$, and $\alpha_1$ and $\alpha_2$ are consistent with all the variables shared between $e_1$ and $e_2$.
    Note that $n' := |V'| = m |\Sigma|^t = m n^{O(\varepsilon \log |\Sigma|)}$ and that $G'$ has a clique of size $m$.

    Let $A'$ be an $n^{-\varepsilon}$-approximation algorithm for the maximum clique problem. Given an instance $I \in \mathcal I$, we first construct a graph $G'$ as above, and then compute a clique $\bm C$ in $G'$ using $A'$.
    We construct an assignment $\bm \sigma$ for $I$ as follows.
    For every variable $v_{e,\alpha}$ in $\bm C$, we set $\bm\sigma(u) = \alpha(u)$ for every $u \in e$.
    By the construction of the graph, $\bm \sigma$ is well defined.
    For $u \in V$ with $\bm \sigma(u)$ undefined, we set $\bm \sigma(u)$ to be an arbitrarily but fixed label.

    We note that $\bm \sigma$ satisfies all the constraints corresponding to edges $e \in E'$ such that $v_{e,\alpha}$ is included in the clique $\bm C$ for some $\alpha \in R_e$.
    This implies that the assignment $\bm \sigma$ satisfies $|\bm C| \geq m/n^\varepsilon$ constraints with probability $1-O(1/n)$, and hence $A$ outputs an $n^{-\varepsilon}$-approximate assignment with probability $1-O(1/n)$.

    Next, we analyze the sensitivity of $A$.
    When we delete a constraint in $I$, we change at most $d 2^t = O(n^{2\varepsilon})$ edges in $G'$.
    Hence, the sensitivity of $A$ is at most $ O(n^{2\varepsilon})$ times the sensitivity of $A'$.
    By choosing $\varepsilon \ll \delta$, the sensitivity of $A'$ is
    \[
        \Omega\qty(\frac{n^\delta}{n^{2\varepsilon}})
        = 
        \Omega\qty(n^{\delta/2}) = \Omega\qty((n')^{\delta/O(\varepsilon \log |\Sigma|+1)}),
    \]
    which implies the claim.
\end{proof}

A vertex set $S \subseteq V$ is called an \emph{independent set} if no two vertices $u,v \in S$ have an edge between them.
In the \emph{maximum independent set problem}, given a graph $G = (V, E)$, the goal is to find an independent set of maximum size.
Because a vertex set is an independent set if and only it is a clique in the complement graph, we obtain the following:
\begin{corollary}\label{cor:independent-set}
    There exist universal constants $\varepsilon,\delta >0$ such that any algorithm for the maximum independent set that outputs an $n^{-\varepsilon}$-approximate solution with probability $1-O(1/n)$ has sensitivity $\Omega(n^{\delta})$.
\end{corollary}

\section{Max CSPs}\label{sec:max-csps}

In this section, we show lower bounds for various Max CSPs.
First, we consider $\mathsf{E3SAT}$, which is a special case of Boolean CSPs, where each constraint is a disjunctions of exactly $3$ literals.
\begin{theorem}\label{thm:3SAT}
    There exist universal constants $\varepsilon,\delta>0$ such that any algorithm for $\mathsf{E3SAT}_{1,1-\varepsilon}$ has sensitivity $\Omega(n^\delta)$.
\end{theorem}
\begin{proof}
    Let $\varepsilon,\delta>0$ and $d,k \geq 1$ be as in \Cref{thm:label-cover}.
    That is, any algorithm for $\mathsf{LabelCover}_{1,1-\varepsilon}$ over a bipartite graph of maximum degree $d$ and over a domain of size $k$ has sensitivity of $\Omega(n^\delta)$.

    We consider the following transformation $T_I$ from label cover instances to SAT instances.
    Let $I=(U,V,E,\Sigma_U,\Sigma_V, \mathcal R)$ be a satisfiable label cover instance with $|\Sigma_V| \leq |\Sigma_U| = k$.
    We encode each label using $O(\log k)$ bits and introduce corresponding Boolean variables. As well, we introduce a constraint simulating $R_e$ for each $e \in E$ over the corresponding Boolean variables. Let the Boolean CSP instance that results from this procedure by denoted by $\hat I = (\hat V,\hat E, \{0,1\}, \hat{\mathcal R})$, and note that each hyperedge $e \in \hat E$ is of size at most $O(\log k)$.
    Then, we further transform it to a SAT instance $I' = (V',E',\{0,1\}, \mathcal R')$ by converting each constraint in $\hat{\mathcal R}$ to an E3CNF formula.
    Note that we may need to add some auxiliary variables in this step and hence $\hat V$ is a subset of $V'$.
    Without loss of generality we may assume that each E3CNF formula has exactly $K$ clauses, where $K = O(\log k) \cdot 2^{O(\log k)} = O(k\log k)$.    

    Suppose $\sigma':V' \to \{0,1\}$ is an assignment for $I'$.
    Then, there is a natural transformation $T_\sigma$ that decodes an assignment $\sigma$ for $I$ from $\sigma'$.
    We now show that the pair $(T_I,T_\sigma)$ is a $(1,1-\varepsilon,c'=1,s'=1-\varepsilon/K,C_I=K,C_\sigma=1)$-sensitivity-preserving reduction.
    Then, the claim follows by \Cref{lem:sensitivity-reduction}.
    The analysis for $c'$ and $C_\sigma$ is immediate.

    First, we analyze $s'$.
    Suppose that a (possibly random) assignment $\bm \sigma':V' \to \{0,1\}$ satisfies $\E \mathsf{val}_{I'}(\bm \sigma') \geq 1 - \varepsilon'$.
    This implies that $\bm \sigma'$ violates at most $\varepsilon' m'$ constraints in $I'$ in expectation, and hence $\bm \sigma'$ (restricted to $\hat V$) violates at most $\varepsilon' K \hat m$ constraints in $\hat I$ in expectation.
    It follows that $\bm \sigma$ violates at most $\varepsilon' K \hat m = \varepsilon' K m$ constraints, and hence $\E \mathsf{val}_I(\bm \sigma) \geq 1-\varepsilon'K$.
    Then, $s' = 1-\varepsilon/K$ satisfies \Cref{item:reduction-soundness} of \Cref{def:sensitivity-reduction}.

    Next, we analyze $C_I$.
    Let $I$ and $\tilde I$ be two instances of $\mathsf{LabelCover}_{1,1-\varepsilon}$, where $\tilde I$ is obtained from $I$ by deleting one constraint.
    Then, $I'$ and $\tilde I'$ differ by at most $K$ constraints, and hence we can set $C_I=K=O(k\log k)$.
\end{proof}

Next, we consider $\mathsf{3LIN}$, which is a special type of a Boolean CSP, where we identify the domain $\{0,1\}$ with the group $\mathbb Z_2$ and each constraint is of the form $x=0$, $x=1$, $x+y=0$, $x+y=1$, $x+y+z = 0$, or $x+y+z=1$. 
For $1 \geq c \geq s \geq 0$, we define $\mathsf{Max3LIN}_{c,s}$ denote the problem that, given a $c$-satisfiable instance of $\mathsf{3LIN}$, the goal is to compute an $s$-satisfying assignment.
\begin{corollary}\label{cor:3LIN}
    There exist universal constants $\varepsilon,\delta>0$ such that any algorithm for $\mathsf{Max3LIN}_{4/7,(1-\varepsilon)4/7}$ has sensitivity $\Omega(n^\delta)$.    
\end{corollary}
\begin{proof}
    Consider the following transformation $T_I$ from an instance $I=(V,E,\{0,1\},\mathcal R)$ of $\mathsf{E3SAT}$ to an instance $I'=(V',E',\mathbb Z_2,\mathcal R')$ of $\mathsf{3LIN}$:
    For each constraint of the form $(l_1 \vee l_2 \vee l_3)$, where $l_1,l_2,l_3$ are literals, we introduce seven constraints of the form $l_1=1$, $l_2=1$, $l_3=1$, $l_1+l_2=1$, $l_2+l_3=1$, $l_1+l_3=1$, and $l_1+l_2+l_3=1$, where we identify a negative literal $\bar x$ with $1+ x \pmod 2$.
    In particular, $|V| = |V'|$ holds.

    Given an assignment $\sigma':V' \to \mathbb Z_2$ for $I'$, we simply output the corresponding assignment $\sigma:V \to \{0,1\}$.
    It is easy to confirm that the pair $(T_I,T_\sigma)$ is a $(1,1-\varepsilon,4/7,(1-\varepsilon) \cdot 4/7,7,1)$-sensitivity-preserving reduction.
    Hence, the claim follows by \Cref{lem:sensitivity-reduction}.
\end{proof}

Finally, we consider the maximum cut problem, where we are given a graph $G=(V,E)$, and the goal is to find a set $S\subseteq V$ that maximizes the cut size, i.e., the number of edges between $S$ and $V\setminus S$.
Using the reduction in~\cite{trevisan2000gadgets}, we obtain the following sensitivity lower bound for the maximum cut problem.
\begin{corollary}\label{cor:max-cut}
    There exist universal constants $c,\varepsilon,\delta>0$ such that any algorithm for $\mathsf{MaxCut}_{c,c(1-\varepsilon)}$ has sensitivity $\Omega(n^\delta)$.
\end{corollary}

\section{Non-Signaling Model}\label{sec:distributed}

We consider the \emph{non-signaling model}~\cite{akbari2025online,arfaoui2014can,gavoille2009can}, where, given a graph $G=(V,E)$ and a graph problem, we can produce an arbitrary output distribution $\mu_G$ as long as it does not violate the non-signaling principle: for any set of vertices $S\subseteq V$, modifying the structure of the input graph at more than a distance $t$ from $S$ does not affect the output distribution of $S$.
The parameter $t$ is called the \emph{locality} of the model.
In this section, we study the relationship between the non-signaling model and sensitivity, and we derive several lower bounds for the non-signaling model.
We note that these lower bounds immediately imply lower bounds for the $\mathsf{LOCAL}$ model, the quantum-$\mathsf{LOCAL}$ model, and the bounded dependence model.

A graph optimization problem $\mathcal{P}$ is specified by a family of objective functions $\{f_{\mathcal{P},G} : \mathcal{X}_G \to \mathbb{R}\}_{G \in \mathcal{G}}$, where $\mathcal{G}$ is the class of all graphs and $\mathcal{X}_G$ is the set of feasible solutions for the graph $G$.
Given a graph $G = (V,E)$, the task is to output a (possibly randomized) feasible solution $X \in \mathcal{X}_G$ that maximizes (or minimizes) the expected value $\E[f_{\mathcal{P},G}(X)]$.

The following observation shows that if there exists a distribution in the non-signaling model with small locality, then it can be transformed into an algorithm with low sensitivity.
\begin{theorem}\label{thm:distributed-to-sensitivity}
    Let $\mathcal P$ be a graph optimization problem, where the output is a vertex set.
    If there is a non-signaling distribution for $\mathcal P$ with locality $t$ on graphs with maximum degree at most $\Delta$, then there exists an algorithm for $\mathcal P$ with sensitivity $O(\Delta^t)$.
\end{theorem}
\begin{proof}
    Let $\{\mu_G\}_{G \in \mathcal G}$ be the distributions for $\mathcal P$ in the non-signaling model with locality $t$.
    Then, we consider the algorithm $A$ that, given a graph $G=(V,E)$, we sample a solution from $\mu_G$ and then output it.

    We now analyze the sensitivity of $A$.
    Let $G=(V,E)$ be a graph and consider $\tilde G=G-e$ for some $e \in E$.
    Let $S$ be the set of vertices with distance at least $t+1$ from the endpoints of $e$.
    From the non-signaling principle, the marginal distributions of $\mu_G$ and $\mu_{\tilde G}$ on $S$ are exactly the same.
    Because the number of vertices $v \in V$ such that $e$ belongs to the $t$-hop neighborhood is bounded by $2\Delta^t$, we have 
    \[
        \EMD(A(G), A(\tilde G)) \leq 2\Delta^t = O(\Delta^t).
        \qedhere
    \]
\end{proof}

\begin{corollary}\label{cor:distributed-to-sensitivity}
    Let $\mathcal P$ be a graph optimization problem, where the output is a vertex set.
    If any algorithm for $\mathcal P$ on graphs with maximum degree at most $\Delta$ has sensitivity at least $s(n)$, then any non-signaling distribution for $\mathcal P$ have locality $\Omega(\log_\Delta s(n))$.
\end{corollary}

Using \Cref{cor:distributed-to-sensitivity}, we can extend various lower bounds that were only known for the $\mathsf{LOCAL}$ model to the non-signaling model.
\begin{description}
    \item[Maximum independent set.] \Cref{cor:independent-set} states that there exist constants $\varepsilon,\delta>0$ such that any algorithm for the maximum independent set problem on graphs with maximum degree $n^{O(\varepsilon)}$ that outputs an $n^{-\varepsilon}$-approximate solution with probability $1-O(1/n)$ has sensitivity $\Omega(n^\delta)$.
    Hence by \Cref{cor:distributed-to-sensitivity}, any non-signaling  distribution for the maximum independent set problem that outputs an $n^{-\varepsilon}$-approximation solution with probability $1-O(1/n)$ has locality
    \[
        \Omega\qty(\log_{n^{O(\varepsilon)}} n^\delta) = \Omega\qty(\frac{1}{\varepsilon}).
    \]
    This extends the same lower bound on round complexity for the $\mathsf{LOCAL}$ model~\cite{bodlaender2016brief}. 

    \item[Minimum vertex cover.]
    \Cref{cor:vertex-cover} states that there exist $\varepsilon,\delta>0$ such that any $(1+\varepsilon)$-approximation algorithm for the minimum vertex cover problem on bounded-degree graphs has sensitivity  $\Omega(n^\delta)$.
    Hence by \Cref{cor:distributed-to-sensitivity}, the locality of any non-signaling distribution that outputs a  $(1+\varepsilon)$-approximation for the minimum vertex cover problem must be $\Omega(\log n^\delta) = \Omega(\log n)$.
    This extends the same lower bound on round complexity for the $\mathsf{LOCAL}$ model~\cite{goos2014no,faour2022distributed}.

    \item[Maximum cut.]  \Cref{cor:max-cut} states that there exist $\varepsilon,\delta>0$ such that any $(1+\varepsilon)$-approximation algorithm for the maximum cut problem on bounded-degree graphs has sensitivity $\Omega(n^\delta)$.
    Hence by \Cref{cor:distributed-to-sensitivity}, the locality of any non-signaling distribution that outputs a $(1-\varepsilon)$-approximation algorithm for the maximum cut problem must be $\Omega(\log n^\delta) = \Omega(\log n)$.
    This extends the same lower bound on round complexity for the $\mathsf{LOCAL}$ model~\cite{chang2023complexity}.
\end{description}

\section*{Acknowledgments} 
We thank an anonymous reviewer for pointing out the connection to distributed models of computation beyond the $\mathsf{LOCAL}$ model, as well as for their thorough review of our work.

\bibliographystyle{alphaurl}
\newcommand{\etalchar}[1]{$^{#1}$}

\section{Why Simple Union-Based Constructions Fail}\label{sec:appendix}

In this appendix we describe why the simpler approach of taking the union of two families of instances---one requiring high sensitivity and one which is inapproximable---fails. Suppose that there are two families of instances $\mathcal I_n$ and $\mathcal J_n$ with $n$ variables and $m = m(n)$ constraints where $\mathrm{Opt}(I)=\mathrm{Opt}(J)=m$ for each $I \in \mathcal I_n$ and $J \in \mathcal J_n$. Suppose further that $\mathcal J = \{\mathcal J_n\}$ is a computationally hard family of instances, meaning that for any polynomial-time algorithm $A$ (and for infinitely many $n$) there is a $J \in \mathcal J_n$ such that $A$ can satisfy at most a $1-\varepsilon$ fraction of the constraints of $J$. As well, suppose that we have a lower bound on the sensitivity of exactly computing $\mathcal I = \{\mathcal I_n\}$, such as the one from Section 4. Define $\mathcal{I}_n \cup \mathcal{J}_n := \{ I \cup J : I \in \mathcal{I}_n, J \in \mathcal{J}_n \}$. Note that a $(1-\varepsilon/2)$-approximation algorithm for $\mathcal I \cup \mathcal J = \{\mathcal{I}_n \cup \mathcal{J}_n\}$ satisfies at least $2m(1-\varepsilon/2) = m + m(1-\varepsilon)$-many constraints of any instance $I \cup J \in \mathcal I \cup \mathcal J$, where $I$ and $J$ each have $m$ constraints. One would like to argue that any $(1-\varepsilon/2)$-approximation algorithm $A$ for $\mathcal{I} \cup \mathcal{J}$ which has low sensitivity would imply that there exists an algorithm with low sensitivity which computes $\mathcal I$ exactly, contradicting our sensitivity lower bound for $\mathcal I$. We will explain why this is not obviously the case; in particular why it is not clear that we can force an algorithm to output a satisfying assignment to $I$.

Any algorithm $A$ for $\mathcal{I} \cup \mathcal{J}$ implies an algorithm for $\mathcal{J}$: on input $J \in \mathcal{J}$, generate an input $I \in \mathcal{I}$ and output the output of $A(I \cup J)$ restricted to the variables of $J$; call this algorithm $A_I$ and note that there is one for each $I \in \mathcal{I}$. The computational hardness assumption for $\mathcal{J}$ then implies that (for infinitely many $n$) there exists some $J_{A_I} \in \mathcal J_n$ such that $A_I$ outputs at most a $(1-\varepsilon)$-approximate solution on $J_{A_I}$. Hence, as $A$ is a $(1-\varepsilon/2)$ approximating algorithm, and $A_I(J)$ satisfies at most $m(1-\varepsilon)$-many constraints, $A(I \cup J)$ must compute a satisfying assignment to $I$.

Now, we would like to say that our sensitivity lower bound for computing $\mathcal{I}$ exactly implies a sensitivity lower bound for approximating $\mathcal{I} \cup \mathcal{J}$. However, it is not clear that the sensitivity lower bound lifts. In particular, consider two neighboring instances $I,I’ \in \mathcal I$. We would like to show a sensitivity lower bound against neighboring instances for $\mathcal{I} \cup \mathcal{J}$. However, $J_{A_I}$ and $J_{A_{I'}}$ may be different instances (far from neighboring). Furthermore, while we know that $J_{A_I}$ is hard for $A_I$ and $J_{A_{I'}}$ is hard for $A_{I'}$, it may not be the case that $J_{A_I}$ is hard for $A_{I'}$. In particular, $A(I \cup J_{A_I})$ must compute a satisfying assignment for $I$ as it can only satisfy $m(1-\varepsilon)$-many constraints of $I$. However, $A(I' \cup J_{A_I})$ could satisfy many more constraints of $J_{A_I}$ and hence is not required to compute a satisfying assignment on $I'$ (it could violate some constraints of $I'$), and so the sensitivity bound against computing satisfying assignments does not apply.  

\end{document}